\documentclass[submission,copyright,creativecommons,sharealike,noncommercial]{eptcs}
\providecommand{\event}{QPL 2017}
\pdfoutput=1

\usepackage{mathtools} 
\usepackage{amssymb} 
\usepackage{bbold} 
\usepackage{amsthm} 
\usepackage{stmaryrd} 

\usepackage{relsize} 
\usepackage{microtype} 
\usepackage{multicol} 
\usepackage{csquotes} 

\usepackage{hyperref} 
\usepackage[nocompress]{cite} 

\usepackage{graphicx} 
\usepackage[usenames,dvipsnames]{xcolor} 
\usepackage{tikz} 
\usepackage{circuitikz} 
\usetikzlibrary{
	arrows,
	shapes,
	decorations,
	intersections,
	backgrounds,
	positioning,
	circuits.ee.IEC
	}

		\newcounter{theorem_c} 
		\numberwithin{theorem_c}{section} 
		\numberwithin{equation}{section} 

		\theoremstyle{plain} 
		\newtheorem{theorem}[theorem_c]{Theorem}

		\newtheoremstyle{exampstyle}
		  {2mm} 
		  {2mm} 
		  {\itshape} 
		  {} 
		  {\bfseries} 
		  {.} 
		  {.5em} 
		  {} 

		\theoremstyle{exampstyle}

		\newtheorem{remark}[theorem_c]{Remark}

	\newcommand{\inlineQuote}[1]{\textquotedblleft #1\textquotedblright} 
	\newcommand{\goodchi}{\protect\raisebox{2pt}{$\chi$}} 

	\newcommand{\naturals}{\mathbb{N}} 
	\newcommand{\integers}{\mathbb{Z}} 
	\newcommand{\circleGroup}{\mathbb{T}} 
	\newcommand{\torusGroup}[1]{\circleGroup^{#1}} 
	\newcommand{\reals}{\mathbb{R}} 
	\newcommand{\complexs}{\mathbb{C}} 
	\newcommand{\integersMod}[1]{\mathbb{Z}_{#1}} 
	\newcommand{\restrict}[2]{\left. #1 \right\vert_{#2}} 
	\newcommand{\nonstd}[1]{\,^\star #1}
	\newcommand{\starNaturals}{\nonstd{\naturals}} 
	\newcommand{\starIntegers}{\nonstd{\integers}} 
	\newcommand{\starComplexs}{\nonstd{\complexs}} 
	\newcommand{\starReals}{\nonstd{\reals}} 

	\newcommand{\suchthat}[2]{\left\{#1 \: \middle\vert \: #2\right\}} 

	\newcommand{\stdpartSym}{\operatorname{st}}
	\newcommand{\stdpart}[1]{\stdpartSym(#1)}
	
	\newcommand{\truncate}[1]{\bar{#1}}
	\newcommand{\liftSym}[1]{\operatorname{lift}_{#1}}
	\newcommand{\lift}[2]{\liftSym{#2}[#1]}
	\newcommand{\torusAlgebra}[1]{{\SpaceH_{\,\torusGroup{#1}}}}
	\newcommand{\integersAlgebra}[1]{{\SpaceH_{\,\integers^{#1}}}}
	\newcommand{\realsAlgebra}[1]{{\SpaceH_{\,\reals^{#1}}}}
	\newcommand{\starIntegersMod}[1]{{\nonstd{\integersMod{#1}}}}
	\newcommand{\starIntegersModPow}[2]{{\nonstd{\integersMod{#1}^{#2}}}}

		\newcommand{\ket}[1]{\vert #1 \rangle} 
		\newcommand{\bra}[1]{\langle #1 \vert} 
		\newcommand{\braket}[2]{\langle #1 \vert #2 \rangle} 
		
		\newcommand{\LtwoSym}{\operatorname{L}^2} 
		\newcommand{\Ltwo}[1]{\LtwoSym[#1]} 

		\newcommand{\SpaceH}{\mathcal{H}} 
		\newcommand{\SpaceG}{\mathcal{G}}
		\newcommand{\SpaceK}{\mathcal{K}}

		\newcommand{\isom}{\cong} 
		\newcommand{\epim}{\twoheadrightarrow} 
		\newcommand{\id}[1]{id_{#1}} 

		\newcommand{\Hom}[3]{\operatorname{Hom}_{\,#1}\left[#2,#3\right]} 

		\newcommand{\HilbCategory}{\operatorname{Hilb}} 
		\newcommand{\sHilbCategory}{\operatorname{sHilb}} 
		\newcommand{\fHilbCategory}{\operatorname{fHilb}} 
		\newcommand{\fdHilbCategory}{\fHilbCategory} 
		\newcommand{\starHilbCategory}{^\star\!\HilbCategory} 
		\newcommand{\starHilbCategoryNearStd}{\starHilbCategory^{(std)}}
		\newcommand{\starsHilbCategory}{\nonstd{\operatorname{\sHilbCategory}}}
		\newcommand{\starsHilbCategoryNearStd}{\starsHilbCategory^{(std)}}

	\newcommand{\Dcolour}{black!80}

	\newcommand{\Xbwcolour}{black!80}

	\newcommand{\Zbwcolour}{white}

	\newcommand{\Ybwcolour}{black!15}

	\newcommand{\Wbwcolour}{black!50}

	\tikzset{
	  rectangle with rounded corners north west/.initial=4pt,
	  rectangle with rounded corners south west/.initial=4pt,
	  rectangle with rounded corners north east/.initial=4pt,
	  rectangle with rounded corners south east/.initial=4pt,
	}
	\makeatletter
	\pgfdeclareshape{rectangle with rounded corners}{
	  \inheritsavedanchors[from=rectangle] 
	  \inheritanchorborder[from=rectangle]
	  \inheritanchor[from=rectangle]{center}
	  \inheritanchor[from=rectangle]{north}
	  \inheritanchor[from=rectangle]{south}
	  \inheritanchor[from=rectangle]{west}
	  \inheritanchor[from=rectangle]{east}
	  \inheritanchor[from=rectangle]{north east}
	  \inheritanchor[from=rectangle]{south east}
	  \inheritanchor[from=rectangle]{north west}
	  \inheritanchor[from=rectangle]{south west}
	  \backgroundpath{
	    \southwest \pgf@xa=\pgf@x \pgf@ya=\pgf@y
	    \northeast \pgf@xb=\pgf@x \pgf@yb=\pgf@y
	    \pgfkeysgetvalue{/tikz/rectangle with rounded corners north west}{\pgf@rectc}
	    \pgfsetcornersarced{\pgfpoint{\pgf@rectc}{\pgf@rectc}}
	    \pgfpathmoveto{\pgfpoint{\pgf@xa}{\pgf@ya}}
	    \pgfpathlineto{\pgfpoint{\pgf@xa}{\pgf@yb}}
	    \pgfkeysgetvalue{/tikz/rectangle with rounded corners north east}{\pgf@rectc}
	    \pgfsetcornersarced{\pgfpoint{\pgf@rectc}{\pgf@rectc}}
	    \pgfpathlineto{\pgfpoint{\pgf@xb}{\pgf@yb}}
	    \pgfkeysgetvalue{/tikz/rectangle with rounded corners south east}{\pgf@rectc}
	    \pgfsetcornersarced{\pgfpoint{\pgf@rectc}{\pgf@rectc}}
	    \pgfpathlineto{\pgfpoint{\pgf@xb}{\pgf@ya}}
	    \pgfkeysgetvalue{/tikz/rectangle with rounded corners south west}{\pgf@rectc}
	    \pgfsetcornersarced{\pgfpoint{\pgf@rectc}{\pgf@rectc}}
	    \pgfpathclose
	 }
	}
	\makeatother

	\tikzset{->-/.style={decoration={markings,mark=at position #1 with {\arrow{>}}},postaction={decorate}}}
	\tikzset{-<-/.style={decoration={markings,mark=at position #1 with {\arrow{<}}},postaction={decorate}}}

	\tikzstyle{every picture}=[baseline=-0.25em,scale=0.5]
	\pgfdeclarelayer{edgelayer}
	\pgfdeclarelayer{nodelayer}
	\pgfsetlayers{edgelayer,nodelayer,main}

	\tikzstyle{box} = [draw,shape=rectangle,inner sep=2pt,minimum height=6mm,minimum width=6mm,fill=white] 
	\tikzstyle{boxlarge} = [draw,shape=rectangle,inner sep=2pt,minimum height=1.5cm,minimum width=8mm,fill=white] 
	\tikzstyle{boxLarge} = [draw,shape=rectangle,inner sep=2pt,minimum height=2cm,minimum width=10mm,fill=white] 
	\tikzstyle{boxsmall} = [draw,shape=rectangle,inner sep=2pt,minimum height=3mm,minimum width=3mm,fill=white] 
	\tikzstyle{dot} = [inner sep=0mm,minimum width=3mm,minimum height=3mm,draw,shape=circle,text depth=-0.1mm]
	\tikzstyle{Zbwdot} = [dot, fill=\Zbwcolour]
	\tikzstyle{Xbwdot} = [dot, fill=\Xbwcolour]
	\tikzstyle{Ybwdot} = [dot, fill=\Ybwcolour]
	\tikzstyle{Wbwdot} = [dot, fill=\Wbwcolour]
	\tikzstyle{antipode} = [boxsmall] 

	\tikzstyle{state} = [draw, rectangle with rounded corners,
	  rectangle with rounded corners north west=8pt,
	  rectangle with rounded corners south west=8pt,
	  rectangle with rounded corners north east=0pt,
	  rectangle with rounded corners south east=0pt,
	,inner sep=2pt,minimum height=6mm,minimum width=6mm,fill=white]
	\tikzstyle{statelarge} = [draw, rectangle with rounded corners,
	  rectangle with rounded corners north west=8pt,
	  rectangle with rounded corners south west=8pt,
	  rectangle with rounded corners north east=0pt,
	  rectangle with rounded corners south east=0pt,
	,inner sep=2pt,minimum height=1.5cm,minimum width=8mm,fill=white]
	\tikzstyle{stateLarge} = [draw, rectangle with rounded corners,
	  rectangle with rounded corners north west=8pt,
	  rectangle with rounded corners south west=8pt,
	  rectangle with rounded corners north east=0pt,
	  rectangle with rounded corners south east=0pt,
	,inner sep=2pt,minimum height=2cm,minimum width=8mm,fill=white]
	\tikzstyle{effect} = [draw, rectangle with rounded corners,
	  rectangle with rounded corners north west=0pt,
	  rectangle with rounded corners south west=0pt,
	  rectangle with rounded corners north east=8pt,
	  rectangle with rounded corners south east=8pt,
	,inner sep=2pt,minimum height=6mm,minimum width=6mm,fill=white]
	\tikzstyle{scalar}=[diamond,draw,inner sep=1pt,font=\small,fill=white]

	\tikzstyle{cdnode}=[fill=white]
	\tikzstyle{labelnode}=[fill=white]
	\tikzstyle{tightlabelnode}=[fill=white,inner sep = 0.1mm]
	\tikzstyle{none}=[inner sep=0pt]
	\tikzstyle{whiteline}=[-, line width=4pt, draw=white]

	\tikzstyle{trace}=[circuit ee IEC,thick,ground,scale=2.5]
	\tikzstyle{cotrace}=[circuit ee IEC,thick,ground,rotate=180,scale=2.5]
	\tikzstyle{upground}=[circuit ee IEC,thick,ground,rotate=90,scale=2.5]
	\tikzstyle{downground}=[circuit ee IEC,thick,ground,rotate=-90,scale=2.5]

	\tikzstyle{doubled} = [line width=1.8pt] 
	
	\tikzstyle{empty diagram}=[draw=gray!40!white,dashed,shape=rectangle,minimum width=1cm,minimum height=1cm]

\title{Towards Quantum Field Theory \\ in Categorical Quantum Mechanics}
\author{
	Stefano Gogioso\\
	University of Oxford \\
	\texttt{stefano.gogioso@cs.ox.ac.uk}
	\and
	Fabrizio Genovese \\
	University of Oxford \\
	\texttt{fabrizio.genovese@cs.ox.ac.uk}
}

\def\titlerunning{Towards Quantum Field Theory in Categorical Quantum Mechanics}
\def\authorrunning{S. Gogioso \& F. Genovese}

\begin{document}

\maketitle

\begin{abstract}
	In this work, we use tools from non-standard analysis to introduce infinite-dimensional quantum systems and quantum fields within the framework of Categorical Quantum Mechanics. We define a dagger compact category $^\ast\operatorname{Hilb}$ suitable for the algebraic manipulation of unbounded operators, Dirac deltas and plane-waves. We cover in detail the construction of quantum systems for particles in boxes with periodic boundary conditions, particles on cubic lattices, and particles in real space. Not quite satisfied with this, we show how certain non-separable Hilbert spaces can also be modelled in our non-standard framework, and we explicitly treat the cases of quantum fields on cubic lattices and quantum fields in real space. 
\end{abstract}

\section{Introduction} 
\label{section_introduction}

The rigorous diagrammatic methods of Categorical Quantum Mechanics \cite{Abramsky2004,Coecke2008,Coecke2011,Backens2014,Coecke2014a,Coecke2016a} have found widely successful application to quantum information, quantum computation and the foundations of quantum theory. Until very recently, however, a major limitation of the framework was its lack of applicability to infinite-dimensional quantum systems, which include many iconic examples from textbook quantum mechanics and quantum field theory. Previous work by the authors \cite{Gogioso2016b} partially overcame this limitation, using non-standard analysis \`{a} la Robinson \cite{Robinson1974} to define a dagger compact category $\starHilbCategory$ of infinite-dimensional separable Hilbert spaces. A non-standard approach was chosen because it enables a consistent mathematical treatment of infinitesimal and infinite quantities, such as those involved in the manipulation of unbounded operators, Dirac deltas, plane waves, and many other gadgets and structures featuring in traditional approaches to quantum mechanics. 

The debate about the physical interpretation of infinitesimals and infinities is as old as calculus itself \cite{Leibniz1684}, and their use always attracts a healthy dose of scepticism. In time, this has lead to an interesting dichotomy, where infinitesimals are used as a quick way to convince oneself of the validity of a statement, but limits are then required for formal justification. Non-standard analysis simply provides a framework to completely replace limits with a consistent algebraic treatment of infinitesimals and infinities. As long as one is willing to assign physical meaning to non-convergent limit constructions---and mainstream quantum mechanics certainly seems to be---there should be little or no problem of physical interpretation.

The original definition of $\starHilbCategory$ from Ref. \cite{Gogioso2016b} featured unital $\dagger$-Frobenius algebras on all objects, the main ingredient of Categorical Quantum Mechanics (CQM) which was missing from the category $\HilbCategory$ of infinite-dimensional Hilbert spaces and bounded operators \cite{Abramsky2012b}. It enabled some first, successful applications of CQM methods to infinite-dimensional quantum systems, but was otherwise somewhat limited: most notably, it could not be applied to the case of unbounded quantum particles on real spaces (the single most important textbook example), nor could it tackle the non-separable Hilbert spaces required for the treatment of quantum fields as will be understood as part of this work.

\newpage
\noindent The limitations of the original definition were self-imposed, aimed at keeping the framework simple and more easily relatable, and did not play any relevant role in most constructions we presented. The starting point of this work, in Section \ref{section_StarHilbBeyond}, is a re-definition of the category $\starHilbCategory$, addressing those unnecessary limitations. Specifically, we remove the requirement that the underlying Hilbert spaces for the objects of $\starHilbCategory$ be standard and separable, and we provide a basis-independent formulation of the objects themselves: aside from making the formalism significantly more powerful, this choice has the categorically pleasing effect of identifying $\starHilbCategory$ as a full subcategory of the Karoubi envelope for the category of non-standard Hilbert spaces and $\starComplexs$-linear maps. 

The rest of this work is dedicated to the explicit construction of quantum systems of interest in a number of traditional applications of quantum mechanics. In Section \ref{section_box} we construct the space of wavefunctions in an $n$-dimensional box with periodic boundary conditions $\torusGroup{n}$, while in Section \ref{section_lattice} we construct the space of wavefunctions on an $n$-dimensional lattice $\integers^n$; both these examples were sketched in the original Ref. \cite{Gogioso2016b}, and are here reproduced in additional detail. In Section \ref{section_reals} we construct the space of unbounded wavefunctions in an $n$-dimensional real space $\reals^n$, which we approximate using an infinite lattice of infinitesimal mesh (a well-tested trick in non-standard analysis \cite{Robinson1974}). For each of these three constructions, we provide a strongly complementary pair corresponding to the position and momentum observables. In Sections \ref{section_fieldLattice} and \ref{section_QFT}, we proceed to treat two cases of non-separable standard Hilbert spaces, exploiting a somewhat surprising fact about exponentials of infinite non-standard integers: in Section \ref{section_fieldLattice} we construct the space of quantum fields on an $n$-dimensional lattice $\integers^n$; in Section \ref{section_QFT}, we once again use an infinite lattice of infinitesimal mesh to construct the space of quantum fields in real space $\reals^n$.

Before moving on, we should remark that the originality of this work and the work of Ref. \cite{Gogioso2016b} does not lie in the application of non-standard methods to conventional quantum theory, for which a rich literature already exists \cite{Farrukh1975a,Ozawa1989,Ojima1993a,Ozawa1997,Yamashita2000,Raab2006}. Rather, it lies in the application of non-standard methods to solve a set of issues---lack of Frobenius algebras, compact closed structure and strongly complementary pairs, to mention just a few---which prevented the algebraic/diagrammatic methods of Categorical Quantum Mechanics from being applied to the infinite-dimensional setting. 
\vspace{-1.5mm}

\section{Redefining \texorpdfstring{$\starHilbCategory$}{Star Hilb}}
\label{section_StarHilbBeyond}
\vspace{-0.5mm}
We define the symmetric monoidal category $\starHilbCategory$ (read: \textbf{Star Hilb}) to have objects in the form of pairs $\SpaceH := (|\SpaceH|, P_\SpaceH)$, where $|\SpaceH|$ is a non-standard Hilbert space\footnote{Note that we dropped the requirement that $|\SpaceH|$ be separable, or even that $|\SpaceH| = \nonstd{V}$ for some standard Hilbert space $V$.} (the \textbf{underlying Hilbert space}) and $P_\SpaceH : |\SpaceH| \rightarrow |\SpaceH|$ is an internal non-standard $\starComplexs$-linear map which satisfies the following requirements.
\begin{itemize}
	\item The map $P_\SpaceH$ is a self-adjoint idempotent (we refer to it as the \textbf{truncating projector}).
	\item There is some family $\ket{e_n}_{n=1}^{D}$ of orthonormal vectors\footnote{Note that we dropped the requirement that $\ket{e_n}_{n=1}^{D}$ be (a non-standard extension of) a standard orthonormal basis.} in $|\SpaceH|$, for some $D \in \starNaturals$, such that:
	\begin{equation}\label{truncatingProjectorResolution}
		P_\SpaceH = \sum_{n=1}^{D} \ket{e_n}\bra{e_n}
	\end{equation}
	The existence of such families is guaranteed by Transfer Theorem. Again by Transfer Theorem, $D$ is the same for all such choices of orthonormal families, and we can consistently define the \textbf{dimension} of $\SpaceH$ to be $\dim{\SpaceH} := D \in \starNaturals$.
\end{itemize}
The morphisms of our re-defined $\starHilbCategory$ take the same form as those given in the original definition (they are internal non-standard $\starComplexs$-linear maps, with the truncating projectors acting as identities):
\begin{equation}\label{morphisms}
\Hom{\starHilbCategory}{\SpaceH}{\SpaceG} := \suchthat{\;P_\SpaceG \circ F \circ P_\SpaceH\;}{\;F:  |\SpaceH| \,\rightarrow\, |\SpaceG| \text{ internal linear map}}.
\end{equation}
This means that $\starHilbCategory$ is now a genuine full subcategory of the Karoubi envelope for the category of non-standard Hilbert spaces and non-standard $\starComplexs$-linear maps. The rest of the construction proceeds exactly as it did in the original definition, and $\starHilbCategory$ is a dagger symmetric monoidal category ($\dagger$-SMC, for short). Morphisms can still be represented as matrices (although no longer in a canonical way), by choosing orthonormal sets which diagonalise the relevant truncating projectors:
\begin{equation}\label{matrixRepresentation}
\truncate{F} := P_\SpaceG \circ F \circ P_\SpaceH = 
\sum_{m=1}^{\dim{\SpaceG}} \sum_{n=1}^{\dim{\SpaceH}} 
\ket{f_{m}} \Big( \bra{f_{m}} F \ket{e_{n}} \Big) \bra{e_{n}}.
\end{equation}
The tensor, symmetric braiding and dagger can be defined as usual by looking at the matrix decomposition, and by Transfer Theorem they are invariant under different choices of diagonalising orthonormal sets. Similarly, unital special commutative $\dagger$-Frobenius algebras can be constructed for all orthonormal bases of an object $\SpaceH$ (i.e. for all orthonormal families diagonalising the truncating projector $P_\SpaceH$). The interested reader is referred to Ref. \cite{Gogioso2016b} for the details of the original constructions, which are unchanged. 

Contrary to the dagger symmetric monoidal structure, the compact closed structure in the original definition was given in terms of the chosen orthonormal basis for each object, and needs to be adapted to our new basis-invariant definition. Consider an object $\SpaceH$ of $\starHilbCategory$, together with a diagonalisation $P_\SpaceH = \sum_{n=1}^{\dim{\SpaceH}} \ket{e_n}\bra{e_n}$ of its truncating projector. Also consider the dual $|\SpaceH|^\ast$ to the underlying Hilbert space $|\SpaceH|$ of $\SpaceH$ (which exists by Transfer Theorem), together with the orthonormal set $\ket{\xi_n}_{n=1}^{\dim{\SpaceH}}$ of vectors in $|\SpaceH|^\ast$ specified by the adjoints of the states\footnote{I.e. $\ket{\xi_n}$ is the vector in $|\SpaceH|^\ast$ specified by the linear operator $\bra{e_n}$ on $|\SpaceH|$.} in the orthonormal family $\ket{e_n}_{n=1}^{\dim{\SpaceH}}$. We define the \textbf{dual object} $\SpaceH^\ast$ to be given by the pair $\SpaceH^\ast:=(|\SpaceH|^\ast, P_{\SpaceH^\ast})$, where the truncating projector is defined by $P_{\SpaceH^\ast}:=\sum_{n=1}^{\dim{\SpaceH}}\ket{\xi_n}\bra{\xi_n}$ (it is a standard check that this definition is basis-invariant). The \textbf{compact closed structure} is defined as follows (once again, it is a standard check that this definition is basis-invariant):
\begin{equation}
\begin{tikzpicture}
	\begin{pgfonlayer}{nodelayer}
		\node [style=none] (0) at (-12.75, 0) {};
		\node [style=tightlabelnode] (1) at (-6, 0) {$\sum\limits_{n=1}^{\dim{\SpaceH}} \ket{\xi_n} \otimes \ket{e_n}$};
		\node [style=none, doubled] (2) at (-11.5, 1) {};
		\node [style=none, doubled] (3) at (5, 0) {};
		\node [style=none] (4) at (3.75, -1) {};
		\node [style=none, doubled] (5) at (3.75, 1) {};
		\node [style=none, doubled] (6) at (-11.5, -1) {};
		\node [style=none, doubled] (7) at (-10, 0) {$:=$};
		\node [style=none, doubled] (8) at (6.5, 0) {$:=$};
		\node [style=tightlabelnode] (9) at (10.5, 0) {$\sum\limits_{n=1}^{\dim{\SpaceH}} \bra{e_n} \otimes \bra{\xi_n}$};
	\end{pgfonlayer}
	\begin{pgfonlayer}{edgelayer}
		\draw [style=<-, in=180, out=90] (0.center) to (2.center);
		\draw [style=-, in=-90, out=0] (4.center) to (3.center);
		\draw [style=->, in=90, out=0] (5.center) to (3.center);
		\draw [style=-, in=180, out=-90] (0.center) to (6.center);
	\end{pgfonlayer}
\end{tikzpicture}
\end{equation}
The dual object $\SpaceH^\ast$ has the same dimension as $\SpaceH$, and the non-standard natural number $\dim{\SpaceH}$ coincides with the scalar given by the definition of dimension in dagger compact categories:
\begin{equation}
\begin{tikzpicture}
	\begin{pgfonlayer}{nodelayer}
		\node [style=none] (0) at (-12.5, 0) {};
		\node [style=tightlabelnode] (1) at (-4, 0) {$\sum\limits_{n=1}^{\dim{\SpaceH}} \braket{\xi_n}{\xi_n} \otimes \braket{e_n}{e_n}$};
		\node [style=none, doubled] (2) at (-11.5, 1) {};
		\node [style=none, doubled] (3) at (-10.5, 0) {};
		\node [style=none] (4) at (-11.5, -1) {};
		\node [style=none, doubled] (5) at (-11.5, 1) {};
		\node [style=none, doubled] (6) at (-11.5, -1) {};
		\node [style=none, doubled] (7) at (-9, 0) {$=$};
		\node [style=none, doubled] (8) at (1, 0) {$=$};
		\node [style=tightlabelnode] (9) at (3.5, 0) {$\dim{\SpaceH}$};
	\end{pgfonlayer}
	\begin{pgfonlayer}{edgelayer}
		\draw [style=<-, in=180, out=90] (0.center) to (2.center);
		\draw [style=->, in=-90, out=0] (4.center) to (3.center);
		\draw [style=-, in=90, out=0] (5.center) to (3.center);
		\draw [style=-, in=180, out=-90] (0.center) to (6.center);
	\end{pgfonlayer}
\end{tikzpicture}
\end{equation}

Let $\starsHilbCategory$ be the full subcategory of $\starHilbCategory$ given by those objects $\SpaceH$ such that $|\SpaceH| = \nonstd{V}$ for some separable standard Hilbert space $V$, and such that the truncating projector spans all near-standard vectors. Let $\starsHilbCategoryNearStd$ be the sub-$\dagger$-SMC of $\starsHilbCategory$ given by near-standard morphisms. In particular, the original $\starHilbCategory$ from Ref. \cite{Gogioso2016b} is a full subcategory of the newly defined $\starsHilbCategory$, and the original $\starHilbCategoryNearStd$ featuring in Theorem 3.4 of Ref. \cite{Gogioso2016b} is a full subcategory of the newly defined $\starsHilbCategoryNearStd$. We can define a \textbf{standard part functor} $\stdpartSym : \; \starsHilbCategoryNearStd \rightarrow \sHilbCategory$, which acts as $\SpaceH \mapsto |\SpaceH|$ on objects and as $F \mapsto \stdpart{F}$ on morphisms. The standard part functor is $\complexs$-linear, and identifies two near-standard maps $F,G: \SpaceH \rightarrow \SpaceK$ if and only if $F-G$ has infinitesimal operator norm; this defines an equivalence relation on morphisms in $\starsHilbCategoryNearStd$, which we denote by $\sim$ and refer to as \textbf{infinitesimal equivalence}. The equivalence relation~$\sim$ respects composition, tensor product and dagger, and endows $\starsHilbCategoryNearStd$ with the structure of a $\dagger$-symmetric monoidal 2-category. We can also define a weak\footnote{By weak we mean that composition and tensor product are respected only up to infinitesimal equivalence, i.e. that we have $\lift{f \circ g}{\omega} \sim \lift{f}{\omega} \circ \lift{g}{\omega}$ and $\lift{f \otimes g}{\omega} \sim \lift{f}{\omega} \otimes \lift{g}{\omega}$} \textbf{truncation functor} $\liftSym{\omega}: \sHilbCategory \rightarrow\, \starsHilbCategoryNearStd$, which acts as $V \mapsto (V,P^{(V)})$ on objects\footnote{The truncating projectors $P^{(V)}$ are defined by appropriately choosing an orthonormal basis $\ket{e_n^{(V)}}_n$ for each separable $V$, and letting $P^{(V)} := \sum_{n=1}^{\dim{V}} \ket{e_n^{(V)}}\bra{e_n^{(V)}}$ (where we set $\dim{V} := \omega$ for infinite-dimensional standard Hilbert spaces $V$).} and sends the standard morphism $f: V \rightarrow W $ to the non-standard morphism $P^{(W)} \circ \nonstd{f} \circ P^{(V)}$ (here $\nonstd{\!f}$ is the non-standard extension of $f$). The following result relates $\sHilbCategory$ and $\starsHilbCategoryNearStd_{\omega}$, the full subcategory of $\starsHilbCategoryNearStd$ spanned by those objects $\SpaceH$ having dimension $\dim{\SpaceH} \in \starNaturals$ which is either a finite natural or the infinite natural $\omega$. 
\begin{theorem}[\textbf{Updated version of Theorem 3.4 from Ref. \cite{Gogioso2016b}}] \hfill\\\label{thm_FundamentalTheorem}
The standard part and truncation functors determine a weak equivalence between $\sHilbCategory$ and $\starsHilbCategoryNearStd_{\omega}$:
\begin{enumerate}
	\itemsep0em
	\item[(i)] $\stdpartSym$ is a full dagger monoidal functor, which is surjective on objects;
	\item[(ii)] $\liftSym{\omega}$ is a weak dagger monoidal faithful functor from a $\dagger$-SMC to a $\dagger$-symmetric monoidal 2-category, which is injective and essentially surjective on objects (its restriction to $\fdHilbCategory$ is strict);
	\item[(iii)] $\stdpartSym$ is a left inverse to $\liftSym{\omega}$; 
	\item[(iv)] for all objects $\SpaceH$ of $\,\starsHilbCategoryNearStd_{\omega}$ there is a canonical unitary isomorphism $\truncate{u}_\SpaceH: \SpaceH \rightarrow \lift{\stdpart{\SpaceH}}{\omega}$, the unique one which satisfies $\stdpart{\truncate{u}_\SpaceH} = \id{\stdpart{\SpaceH}}$;
	\item[(v)] for all morphisms $F: \SpaceH \rightarrow \SpaceG$ in $\starsHilbCategoryNearStd_{\omega}$ we have $\truncate{u}_\SpaceG^\dagger \circ \lift{\stdpart{F}}{\omega} \circ \truncate{u}_\SpaceH \sim F$.
\end{enumerate}
\end{theorem}
\begin{proof}
Essentially the same of Theorem 3.4 from Ref. \cite{Gogioso2016b}, using the fact that the subspace defined by the truncating projector $P_\SpaceH$ contains at least all near-standard vectors in $\SpaceH$.
\end{proof}

\begin{remark}
Unfortunately, the updated version of Theorem 3.4 only covers standard separable Hilbert spaces, while we will see that our re-defined $\starHilbCategory$ allows us to treat some standard non-separable spaces as well. We believe Theorem 3.4 to be likely to extend to some subcategory of standard non-separable spaces, but an exact determination of the necessary and sufficient conditions is left to future work.
\end{remark}

\noindent The essence of Theorem \ref{thm_FundamentalTheorem} is that $\sHilbCategory$ is equivalent to the subcategory $\starsHilbCategoryNearStd_{\omega}$ of $\starHilbCategory$ given by near-standard morphisms, as long as we take care to equate morphisms which are infinitesimally close. The equivalence allows one to prove results about $\sHilbCategory$ by working in $\starHilbCategory$ and taking advantage of the full CQM machinery, according to the following general recipe:
\vspace{-2mm}
\begin{multicols}{2}
\begin{enumerate}
	\itemsep0em
	\item[(i)] start from a morphism in $\sHilbCategory$; 
	\item[(ii)] lift to $\starsHilbCategoryNearStd_{\omega}$ via the lifting functor;
	\item[(iii)] work in $\starHilbCategory$, obtain a result in $\starsHilbCategoryNearStd_{\omega}$; 
	\item[(iv)] descend to $\sHilbCategory$ via the standard part functor. 
\end{enumerate}
\end{multicols}
\vspace{-2mm}
\noindent This procedure is conceptually akin to using the two directions of the transfer theorem to prove results of standard analysis using non-standard methods. When proving \textit{equalities} of morphisms in $\sHilbCategory$, it is in fact sufficient to lift both sides via $\liftSym{\omega}$, and prove the equality in $\starHilbCategory$ without further constraints (this is because both sides will necessarily be lifted to $\starsHilbCategoryNearStd_{\omega}$).

The arbitrary choice of infinite natural $\omega$ in the lifting part of the recipe might seem unnatural at first, as the objects $\lift{V}{\omega}$ and $\lift{V}{\omega'}$ are not isomorphic in $\starHilbCategory$ for different infinite naturals $\omega \neq \omega'$. However, this is not actually an issue: from the perspective of $\sHilbCategory$, the two spaces are equivalent for all intents and purposes, and any proof that can be performed in one space can also be performed in the other. This could be made precise by saying that setting $\Phi_{\omega,\omega'} := (\liftSym{\omega'}\circ \stdpartSym)$ defines a weak equivalence $\Phi_{\omega,\omega'} : \starHilbCategoryNearStd_\omega \rightarrow \starHilbCategoryNearStd_{\omega'}$ such that $\stdpart{\Phi_{\omega,\omega'}(\truncate{F})} = \stdpart{\truncate{F}}$ holds for all morphisms $\truncate{F}$ in $ \starHilbCategoryNearStd_\omega$. From this point of view, the same limiting objects (e.g. Diract deltas, plane-waves and unbounded operators) corresponding to different choices of $\omega$ can be interpreted as incarnations of the same conceptual objects seen at different values for the \inlineQuote{infinite cut-off parameter} $\omega$.

\newpage
\section{Particles in boxes with periodic boundary conditions}
\label{section_box}

Consider a quantum particle in a $n$-dimensional box with periodic boundary conditions: this can equivalently be seen as a quantum particle in an $n$-dimensional torus
, and hence it corresponds to the Hilbert space $\Ltwo{\torusGroup{n}}$ (the $n$-dimensional torus is acted freely and transitively upon by its translation group $\torusGroup{n}$). The momentum eigenstates for the particle form a countable complete orthonormal basis for $\Ltwo{\torusGroup{n}}$, with the eigenstate of momentum $\underline{k} \hbar$ ($\underline{k} \in \integers^n$) given by the following square-integrable function $\torusGroup{n} \rightarrow \complexs$:
\vspace{-0.2mm}
\begin{equation}\label{particleBoxMomentumEigenstates}
\ket{\goodchi_{\underline{k}}} := \underline{x} \rightarrow e^{-i 2 \pi \, \underline{k} \cdot \underline{x}}
\end{equation} 
We take a non-standard extension $\ket{\goodchi_{\underline{k}}}_{k_1,...,k_n=-\omega}^{+\omega}$ of the standard orthonormal basis of momentum eigenstates, where $\omega \in \starNaturals$ is some infinite non-standard natural, and we consider the $\big((2 \omega + 1)^n\big)$-dimensional object $\torusAlgebra{n} := (\nonstd{\Ltwo{\torusGroup{n}}},P_\torusAlgebra{n})$ defined by the following truncating projector:
\vspace{-0.2mm}
\begin{equation}
P_\torusAlgebra{n} := \sum_{k_1=-\omega}^{+\omega} ... \sum_{k_n = -\omega}^{+\omega} \ket{\goodchi_{\underline{k}}}\bra{\goodchi_{\underline{k}}}
\end{equation}  
We take this object to be our model in $\starHilbCategory$ of a quantum particle in an $n$-dimensional box with periodic boundary conditions, and we refer to the family $\ket{\goodchi_{\underline{k}}}_{k_1,...,k_n=-\omega}^{+\omega}$ as the \textbf{momentum eigenstates} for $\torusAlgebra{n}$. Pretty much by definition, this object comes with a unital special commutative $\dagger$-Frobenius algebra corresponding to the \textbf{momentum observable} for the particle:
\vspace{-0.2mm}
\begin{equation}\label{particleBoxMomentumFrobAlgebra}
\begin{tikzpicture}
	\begin{pgfonlayer}{nodelayer}
		\node [style=none, doubled] (0) at (6.5, 0) {};
		\node [style=none, doubled] (1) at (-8, 0) {$:=$};
		\node [style=tightlabelnode] (2) at (14.5, 0) {$\sum\limits_{k_1=-\omega}^{+\omega} ... \sum\limits_{k_n=-\omega}^{+\omega}  \bra{\goodchi_{\underline{k}}}$};
		\node [style=none, doubled] (3) at (10.5, 0) {$:=$};
		\node [style=none] (4) at (-12.5, 0) {};
		\node [style=tightlabelnode] (5) at (-2, 0) {$\sum\limits_{k_1=-\omega}^{+\omega} ... \sum\limits_{k_n=-\omega}^{+\omega}  \ket{\goodchi_{\underline{k}}} \otimes \ket{\goodchi_{\underline{k}}} \otimes \bra{\goodchi_{\underline{k}}}$};
		\node [style=none, doubled] (6) at (-9.5, 0.75) {};
		\node [style=none, doubled] (7) at (-9.5, -0.75) {};
		\node [style=Zbwdot] (8) at (-11, 0) {};
		\node [style=Zbwdot] (9) at (8.25, 0) {};
	\end{pgfonlayer}
	\begin{pgfonlayer}{edgelayer}
		\draw [style=-] (9) to (0.center);
		\draw [style=-] (4.center) to (8);
		\draw [style=-, in=180, out=45] (8) to (6.center);
		\draw [style=-, in=180, out=-45] (8) to (7.center);
	\end{pgfonlayer}
\end{tikzpicture}
\end{equation} 
The labels of the momentum eigenstates can be endowed with the infinite abelian group structure of $(\starIntegersModPow{2\omega+1}{n},\oplus,\underline{0})$, which is defined in full detail in the Appendix, and we can consider the linear extension of the group multiplication and unit:
\vspace{-0.2mm}
\begin{equation}\label{particleBoxPositionFrobAlgebra}
\begin{tikzpicture}
	\begin{pgfonlayer}{nodelayer}
		\node [style=none, doubled] (0) at (10, 0) {};
		\node [style=none, doubled] (1) at (-8, 0) {$:=$};
		\node [style=tightlabelnode] (2) at (14.5, 0) {$\ket{\goodchi_{\underline{0}}}$};
		\node [style=none, doubled] (3) at (12, 0) {$:=$};
		\node [style=none] (4) at (-9.5, 0) {};
		\node [style=tightlabelnode] (5) at (-1, 0) {$\sum\limits_{k_1,h_1=-\omega}^{+\omega} ... \sum\limits_{k_n,h_n=-\omega}^{+\omega}  \ket{\goodchi_{\underline{k}\oplus \underline{h}}} \otimes \bra{\goodchi_{\underline{k}}} \otimes \bra{\goodchi_{\underline{h}}}$};
		\node [style=none, doubled] (6) at (-12.5, 0.75) {};
		\node [style=none, doubled] (7) at (-12.5, -0.75) {};
		\node [style=Xbwdot] (8) at (-11, 0) {};
		\node [style=Xbwdot] (9) at (8.25, 0) {};
	\end{pgfonlayer}
	\begin{pgfonlayer}{edgelayer}
		\draw [style=-] (9) to (0.center);
		\draw [style=-, in=135, out=0] (6.center) to (8);
		\draw [style=-, in=-135, out=0] (7.center) to (8);
		\draw [style=-] (8) to (4.center);
	\end{pgfonlayer}
\end{tikzpicture}
\end{equation} 
It is immediate to check that these two maps form, together with their adjoints, a unital quasi-special commutative $\dagger$-Frobenius algebra, with normalisation factor $(2 \omega + 1)^n$. Furthermore, a result of Ref. \cite{Gogioso2017b} guarantees that $(\hbox{\begin{tikzpicture} [scale=1.2,transform shape] 

\def\deltax{0.3} 
\def\deltay{0.5} 


\node [dot, fill=\Dcolour] (mult) at (0,0) {};

\end{tikzpicture}}\!\!,\hbox{\begin{tikzpicture} [scale=1.2,transform shape] 

\def\deltax{0.3} 
\def\deltay{0.5} 


\node [dot, fill=\Zbwcolour] (mult) at (0,0) {};

\end{tikzpicture}}\!\!)$ is a strongly complementary pair of unital $\dagger$-Frobenius algebras\footnote{The relevant result is Theorem 3.4 of Ref. \cite{Gogioso2017b}, which requires $\hbox{\begin{tikzpicture} [scale=1.2,transform shape] 

\def\deltax{0.3} 
\def\deltay{0.5} 


\node [dot, fill=\Zbwcolour] (mult) at (0,0) {};

\end{tikzpicture}}\!\!$ to have enough classical states (the momentum eigenstates), $(\!\hbox{\begin{tikzpicture} [scale=1.2,transform shape,rotate=-90] 

\def\deltax{0.3} 
\def\deltay{0.5} 


\node (mult_label_inl) at (-\deltax,-\deltay) {};
\node (mult_label_inr) at (+\deltax,-\deltay) {};
\node [dot, fill=\Dcolour] (mult) at (0,0) {};
\node (mult_label_out) at (0,+\deltay) {};

\draw[-] [out=90,in=225](mult_label_inl) to (mult);
\draw[-] [out=90,in=315](mult_label_inr) to (mult);
\draw[-] (mult) to (mult_label_out);

\end{tikzpicture}}\!\!,\!\hbox{\begin{tikzpicture} [scale=1.2,transform shape,rotate=-90] 

\def\deltax{0.3} 
\def\deltay{0.5} 

\path[use as bounding box] (-\deltax,-\deltay) rectangle (\deltax,\deltay);

\node [dot, fill=black] (mult) at (0,0) {};
\node (mult_label_out) at (0,+\deltay) {};
\draw[-] (mult) to (mult_label_out);

\end{tikzpicture}}\!\!)$ to endow the $\hbox{\begin{tikzpicture} [scale=1.2,transform shape] 

\def\deltax{0.3} 
\def\deltay{0.5} 


\node [dot, fill=\Zbwcolour] (mult) at (0,0) {};

\end{tikzpicture}}\!\!$-classical states with the structure of a group (the group $\starIntegersModPow{2\omega+1}{n}$), and $\!\hbox{\begin{tikzpicture} [scale=1.2,transform shape,rotate=-90] 

\def\deltax{0.3} 
\def\deltay{0.5} 

\path[use as bounding box] (-\deltax,-\deltay) rectangle (\deltax,\deltay);

\node [dot, fill=\Zbwcolour] (mult) at (0,-0.25*\deltay) {};
\node (mult_label_out) at (0,+\deltay) {};
\draw[-] (mult) to (mult_label_out);

\end{tikzpicture}}\!\!$ to be $\hbox{\begin{tikzpicture} [scale=1.2,transform shape] 

\def\deltax{0.3} 
\def\deltay{0.5} 


\node [dot, fill=\Dcolour] (mult) at (0,0) {};

\end{tikzpicture}}\!\!$-classical.}. The following \textbf{position eigenstates} $\ket{\delta_{\underline{x}}}$, indexed by $\underline{x} \in \frac{1}{2\omega+1}\starIntegersModPow{2\omega+1}{n}$ (a lattice in the non-standard torus $\nonstd{\torusGroup{n}}$) are the $\hbox{\begin{tikzpicture} [scale=1.2,transform shape] 

\def\deltax{0.3} 
\def\deltay{0.5} 


\node [dot, fill=\Dcolour] (mult) at (0,0) {};

\end{tikzpicture}}\!\!$-classical states (they are orthogonal\footnote{They are orthogonal because they are copied by a quasi-special commutative $\dagger$-FA in a SMC where scalars form a field \cite{Coecke2013b}.} and have square norm $(2 \omega + 1)^n$; see the Appendix for a proof):
\vspace{-0.2mm}
\begin{equation}
	\ket{\delta_{\underline{x}}} := \sum_{k_1 = -\omega}^{+ \omega} ... \sum_{k_n = -\omega}^{+\omega} \goodchi_{\underline{k}}(\underline{x})^\ast \ket{\goodchi_{\underline{k}}}
\end{equation}
For all $\underline{x} \in \frac{1}{2\omega+1}\starIntegersModPow{2\omega+1}{n}$ and all standard smooth $f \in \Ltwo{\torusGroup{n}}$, the position eigenstates satisfy the identity $\stdpart{\braket{\delta_{\underline{x}}}{f}} =  f(\stdpart{\underline{x}})$ by Transfer Theorem: hence the position eigenstates (as we defined them) behave exactly as expected from Dirac delta functions, and we can legitimately refer to $\hbox{\begin{tikzpicture} [scale=1.2,transform shape] 

\def\deltax{0.3} 
\def\deltay{0.5} 


\node [dot, fill=\Dcolour] (mult) at (0,0) {};

\end{tikzpicture}}\!\!$ as the \textbf{position observable} for the particle. Note the duality between the large-scale cut-off on momenta in $\starIntegers^n$ (which are bounded in magnitude by $\sqrt{n} \omega \hbar$) and the small-scale cut-off on positions in $\nonstd{\torusGroup{n}}$ (which are discretised onto a lattice of infinitesimal mesh $\frac{1}{2\omega+1}$): in our non-standard framework, this well-known phenomenon arises in a purely algebraic way from the copy condition for $\hbox{\begin{tikzpicture} [scale=1.2,transform shape] 

\def\deltax{0.3} 
\def\deltay{0.5} 


\node [dot, fill=\Dcolour] (mult) at (0,0) {};

\end{tikzpicture}}\!\!$-classical states (see the Appendix).

\newpage
\section{Particles on lattices}
\label{section_lattice}

Consider a particle on an $n$-dimensional cubic lattice, corresponding to the Hilbert space $\Ltwo{\integers^n}$. The position eigenstates for the particle form a countable complete orthonormal basis for $\Ltwo{\integers^n}$, with the eigenstate of position $\underline{k} \in \integers^n$ given by the following square-integrable function $\integers^n \rightarrow \complexs$:
\begin{equation}
	\ket{\delta_{\underline{k}}} := \underline{h} \mapsto 
	\begin{cases}
		1 &\text{ if } \underline{k} = \underline{h} \\
		0 &\text{ otherwise}
	\end{cases}
\end{equation}  
We take a non-standard extension $\ket{\delta_{\underline{k}}}_{k_1,...,k_n = -\omega}^{\omega}$ of the standard orthonormal basis of position eigenstates, where $\omega \in \starNaturals$ is some infinite non-standard natural, and we consider the $\big((2 \omega + 1)^n\big)$-dimensional object $\integersAlgebra{n} := (\nonstd{\Ltwo{\integers^n}},P_\integersAlgebra{n})$ defined by the following truncating projector:
\begin{equation}
P_\integersAlgebra{n} := \sum_{k_1=-\omega}^{+\omega} ... \sum_{k_n = -\omega}^{+\omega} \ket{\delta_{\underline{k}}}\bra{\delta_{\underline{k}}}
\end{equation}  
We take this object to be our model in $\starHilbCategory$ of a quantum particle on an $n$-dimensional lattice, and we refer to the family $\ket{\delta_{\underline{k}}}_{k_1,...,k_n=-\omega}^{+\omega}$ as the \textbf{position eigenstates} for $\integersAlgebra{n}$. Pretty much by definition, this object comes with a unital special commutative $\dagger$-Frobenius algebra corresponding to the \textbf{position observable} for the particle:
\begin{equation}\label{particleLatticePositionFrobAlgebra}
\begin{tikzpicture}
	\begin{pgfonlayer}{nodelayer}
		\node [style=none, doubled] (0) at (6.5, 0) {};
		\node [style=none, doubled] (1) at (-8, 0) {$:=$};
		\node [style=tightlabelnode] (2) at (14.5, 0) {$\sum\limits_{k_1=-\omega}^{+\omega} ... \sum\limits_{k_n=-\omega}^{+\omega}  \bra{\delta_{\underline{k}}}$};
		\node [style=none, doubled] (3) at (10.5, 0) {$:=$};
		\node [style=none] (4) at (-12.5, 0) {};
		\node [style=tightlabelnode] (5) at (-2, 0) {$\sum\limits_{k_1=-\omega}^{+\omega} ... \sum\limits_{k_n=-\omega}^{+\omega}  \ket{\delta_{\underline{k}}} \otimes \ket{\delta_{\underline{k}}} \otimes \bra{\delta_{\underline{k}}}$};
		\node [style=none, doubled] (6) at (-9.5, 0.75) {};
		\node [style=none, doubled] (7) at (-9.5, -0.75) {};
		\node [style=Xbwdot] (8) at (-11, 0) {};
		\node [style=Xbwdot] (9) at (8.25, 0) {};
	\end{pgfonlayer}
	\begin{pgfonlayer}{edgelayer}
		\draw [style=-] (9) to (0.center);
		\draw [style=-] (4.center) to (8);
		\draw [style=-, in=180, out=45] (8) to (6.center);
		\draw [style=-, in=180, out=-45] (8) to (7.center);
	\end{pgfonlayer}
\end{tikzpicture}
\end{equation} 
The labels of the position eigenstates can be endowed with the infinite abelian group structure of $(\starIntegersModPow{2\omega+1}{n},\oplus,\underline{0})$, and we can consider the linear extension of the group multiplication and unit:
\begin{equation}\label{particleLatticeMomentumFrobAlgebra}
\begin{tikzpicture}
	\begin{pgfonlayer}{nodelayer}
		\node [style=none, doubled] (0) at (10, 0) {};
		\node [style=none, doubled] (1) at (-8, 0) {$:=$};
		\node [style=tightlabelnode] (2) at (14.5, 0) {$\ket{\delta_{\underline{0}}}$};
		\node [style=none, doubled] (3) at (12, 0) {$:=$};
		\node [style=none] (4) at (-9.5, 0) {};
		\node [style=tightlabelnode] (5) at (-1, 0) {$\sum\limits_{k_1,h_1=-\omega}^{+\omega} ... \sum\limits_{k_n,h_n=-\omega}^{+\omega}  \ket{\delta_{\underline{k}\oplus \underline{h}}} \otimes \bra{\delta_{\underline{k}}} \otimes \bra{\delta_{\underline{h}}}$};
		\node [style=none, doubled] (6) at (-12.5, 0.75) {};
		\node [style=none, doubled] (7) at (-12.5, -0.75) {};
		\node [style=Zbwdot] (8) at (-11, 0) {};
		\node [style=Zbwdot] (9) at (8.25, 0) {};
	\end{pgfonlayer}
	\begin{pgfonlayer}{edgelayer}
		\draw [style=-] (9) to (0.center);
		\draw [style=-, in=135, out=0] (6.center) to (8);
		\draw [style=-, in=-135, out=0] (7.center) to (8);
		\draw [style=-] (8) to (4.center);
	\end{pgfonlayer}
\end{tikzpicture}
\end{equation} 
It is immediate to check that these two maps form, together with their adjoints, a unital quasi-special commutative $\dagger$-Frobenius algebra, with normalisation factor $(2 \omega + 1)^n$. Furthermore, a result of Ref. \cite{Gogioso2017b} guarantees that $(\hbox{\begin{tikzpicture} [scale=1.2,transform shape] 

\def\deltax{0.3} 
\def\deltay{0.5} 


\node [dot, fill=\Dcolour] (mult) at (0,0) {};

\end{tikzpicture}}\!\!,\hbox{\begin{tikzpicture} [scale=1.2,transform shape] 

\def\deltax{0.3} 
\def\deltay{0.5} 


\node [dot, fill=\Zbwcolour] (mult) at (0,0) {};

\end{tikzpicture}}\!\!)$ is a strongly complementary pair of unital $\dagger$-Frobenius algebras. 

The infinite abelian group $(\starIntegersMod{2 \omega + 1},\oplus,0)$ is defined in full detail in the Appendix, and has the interval of non-standard integers $\{-\omega,...,+\omega\}$ as its underlying set. Remarkably, it contains the standard integers $\integers$ as a subgroup: as a consequence, the group $(\starIntegersModPow{2\omega+1}{n},\oplus,\underline{0})$ is a legitimate non-standard extension of the translation group $\integers^n$ for the $n$-dimensional lattice. In a sense, we are seeing standard infinite lattices are actually being periodic, but circling around \inlineQuote{beyond standard infinity} rather than at some finite point.

A proof on the same lines of the one mentioned in the previous section shows that the following \textbf{momentum eigenstates} $\ket{\goodchi_{\underline{x}}}$, indexed by $\underline{x} \in \frac{1}{2\omega+1}\starIntegersModPow{2\omega+1}{n}$, are the $\hbox{\begin{tikzpicture} [scale=1.2,transform shape] 

\def\deltax{0.3} 
\def\deltay{0.5} 


\node [dot, fill=\Zbwcolour] (mult) at (0,0) {};

\end{tikzpicture}}\!\!$-classical states (they are orthogonal and have square norm $(2 \omega + 1)^n$):
\begin{equation}
	\ket{\goodchi_{\underline{x}}} := \sum_{k_1 = -\omega}^{+ \omega} ... \sum_{k_n = -\omega}^{+\omega} e^{-i 2 \pi \, \underline{k} \cdot \underline{x}} \ket{\delta_{\underline{k}}}
\end{equation}
From their formulation, it is immediately clear that these are exactly the plane waves, and as a consequence we can legitimately refer to $\hbox{\begin{tikzpicture} [scale=1.2,transform shape] 

\def\deltax{0.3} 
\def\deltay{0.5} 


\node [dot, fill=\Zbwcolour] (mult) at (0,0) {};

\end{tikzpicture}}\!\!$ as the \textbf{momentum observable} for the particle on the lattice. Similarly to the previous Section, we have a nice duality between the large-scale cut-off on positions in $\starIntegers^n$ (which are bounded in magnitude by $\sqrt{n}\omega$) and the small-scale cut-off on momenta in $\nonstd{\torusGroup{n}}$ (which are discretised onto a lattice of infinitesimal mesh $\frac{\hbar}{2\omega+1}$).

\newpage
\section{Particles in real space}
\label{section_reals}

Consider a quantum particle in $n$-dimensional real space, corresponding to the Hilbert space $\Ltwo{\reals^n}$. This is a separable space, but for obvious reasons it does not come with a choice of complete countable orthonormal basis of standard states which is invariant under the translation group $\reals^n$ of the underlying space. However, we can try and approximate one using a subdivision of $\reals^n$ with infinitesimal mesh, as done in Ref. \cite{Robinson1974} to model Riemann integration (as well as differentiation and a number of other fundamental constructions of calculus).

We fix two infinite natural numbers: an \textbf{infrared (IR) infinity} $\omega_{ir}$, which will govern the large scale limit of our approximation, and a \textbf{ultraviolet (UV) infinity} $\omega_{uv}$, which will govern the small scale limit. We require both to be odd, and we set $2 \omega + 1 := \omega_{uv}\omega_{ir}$. Consider the following internal lattice---isomorphic to $\starIntegersModPow{2\omega+1}{n}$, of which it inherits the group structure---in $n$-dim non-standard real space $\starReals^n$:
\begin{equation}
\dfrac{1}{\omega_{uv}} \starIntegersModPow{2\omega +1}{n} := \suchthat{\underline{p} \in \starReals^n}{ \underline{p} = \dfrac{\underline{k}}{\omega_{uv}},\, \underline{k} \in \starIntegersModPow{2\omega+1}{n} } \subset \starReals^n
\end{equation}
Note that the group $\reals^n$ is not a subgroup of $\frac{1}{\omega_{uv}}\starIntegersModPow{2\omega+1}{n}$, as the elements of the latter are non-standard rational numbers. However the former can be approximated by the latter in the following rigorous sense: there is a subgroup $\big(\frac{1}{\omega_{uv}}\starIntegersModPow{2\omega+1}{n}\big)_0 \triangleleft \frac{1}{\omega_{uv}}\starIntegersModPow{2\omega+1}{n}$ given by taking the near standard elements, and taking a quotient of this subgroup through the standard part function yields $\reals^n$ as the quotient group (see Appendix). We can interpret the group $\frac{1}{\omega_{uv}}\starIntegersModPow{2\omega+1}{n}$ as a lattice with infinitesimally fine mesh (specified by $\frac{1}{\omega_{uv}}$) which approximates $\reals^n$, which covers it entirely, and which circles around \inlineQuote{beyond standard infinity} (at a point specified by $\omega_{ir}$). Having seen this, we consider the following orthonormal family of non-standard functions $\nonstd{\reals^n} \rightarrow \starComplexs$, indexed by the points $\underline{p} \in \frac{1}{\omega_{uv}}\starIntegersModPow{2\omega+1}{n}$ of the lattice which we introduced above::
\begin{equation}
	\ket{\goodchi_{\underline{p}}} := \underline{x} \mapsto \frac{1}{\sqrt{\omega_{uv}}}e^{-i 2\pi \, (\underline{p}\cdot \underline{x})}
\end{equation} 
Just like the family of momentum eigenstates defined in Section \ref{section_box}, this family is orthonormal, and hence we can take it to define a $\big((2 \omega + 1)^n\big)$-dimensional object $\realsAlgebra{n} := (\nonstd{\Ltwo{\reals^n}},P_\realsAlgebra{n})$ of $\starHilbCategory$, by considering the following truncating projector:
\begin{equation}
P_\realsAlgebra{n} := \sum_{p_1=-\omega_{ir}}^{+\omega_{ir}} ... \sum_{p_n = -\omega_{ir}}^{+\omega_{ir}} \ket{\goodchi_{\underline{p}}}\bra{\goodchi_{\underline{p}}} \;  :\equiv \sum_{k_1=-\omega}^{+\omega} ... \sum_{k_n = -\omega}^{+\omega} \ket{\goodchi_{\underline{k}/\omega_{uv}}}\bra{\goodchi_{\underline{k}/\omega_{uv}}}
\end{equation}  
The shorthand notation $\sum_{p=-\omega_{ir}}^{\omega_{ir}} f(p) :\equiv \sum_{k = -\omega}^{\omega} f(k / \omega_{uv})$ adopted above for summation over elements of $\frac{1}{\omega_{uv}}\starIntegersModPow{2\omega+1}{n}$ will be used through the rest of this work. We take this object to be our model in $\starHilbCategory$ of an unbounded quantum particle in an $n$-dimensional real space, and we refer to the family $\ket{\goodchi_{\underline{p}}}_{p_1,...,p_n=-\omega_{ir}}^{+\omega_{ir}}$ as the \textbf{momentum eigenstates} for $\realsAlgebra{n}$. Pretty much by definition, this object comes with a unital special commutative $\dagger$-Frobenius algebra corresponding to the \textbf{momentum observable} for the unbounded particle:
\begin{equation}\label{particleUnboundedMomentumFrobAlgebra}
\begin{tikzpicture}
	\begin{pgfonlayer}{nodelayer}
		\node [style=none, doubled] (0) at (6.5, 0) {};
		\node [style=none, doubled] (1) at (-8, 0) {$:=$};
		\node [style=tightlabelnode] (2) at (14.5, 0) {$\sum\limits_{p_1=-\omega_{ir}}^{+\omega_{ir}} ... \sum\limits_{p_n=-\omega_{ir}}^{+\omega_{ir}}  \bra{\goodchi_{\underline{p}}}$};
		\node [style=none, doubled] (3) at (10.5, 0) {$:=$};
		\node [style=none] (4) at (-12.5, 0) {};
		\node [style=tightlabelnode] (5) at (-2, 0) {$\sum\limits_{p_1=-\omega_{ir}}^{+\omega_{ir}} ... \sum\limits_{p_n=-\omega_{ir}}^{+\omega_{ir}}  \ket{\goodchi_{\underline{p}}} \otimes \ket{\goodchi_{\underline{p}}} \otimes \bra{\goodchi_{\underline{p}}}$};
		\node [style=none, doubled] (6) at (-9.5, 0.75) {};
		\node [style=none, doubled] (7) at (-9.5, -0.75) {};
		\node [style=Zbwdot] (8) at (-11, 0) {};
		\node [style=Zbwdot] (9) at (8.25, 0) {};
	\end{pgfonlayer}
	\begin{pgfonlayer}{edgelayer}
		\draw [style=-] (9) to (0.center);
		\draw [style=-] (4.center) to (8);
		\draw [style=-, in=180, out=45] (8) to (6.center);
		\draw [style=-, in=180, out=-45] (8) to (7.center);
	\end{pgfonlayer}
\end{tikzpicture}
\end{equation} 
The labels of the momentum eigenstates come endowed with the infinite abelian group structure of $(\frac{1}{\omega_{uv}}\starIntegersModPow{2\omega+1}{n},\oplus,\underline{0})$, and we can consider the linear extension of the group multiplication and unit:
\begin{equation}\label{particleUnboundedPositionFrobAlgebra}
\begin{tikzpicture}
	\begin{pgfonlayer}{nodelayer}
		\node [style=none, doubled] (0) at (10, 0) {};
		\node [style=none, doubled] (1) at (-8, 0) {$:=$};
		\node [style=tightlabelnode] (2) at (14.5, 0) {$\ket{\goodchi_{\underline{0}}}$};
		\node [style=none, doubled] (3) at (12, 0) {$:=$};
		\node [style=none] (4) at (-9.5, 0) {};
		\node [style=tightlabelnode] (5) at (-1, 0) {$\sum\limits_{p_1,q_1=-\omega_{ir}}^{+\omega_{ir}} ... \sum\limits_{p_n,q_n=-\omega_{ir}}^{+\omega_{ir}}  \ket{\goodchi_{\underline{p}\oplus \underline{q}}} \otimes \bra{\goodchi_{\underline{p}}} \otimes \bra{\goodchi_{\underline{q}}}$};
		\node [style=none, doubled] (6) at (-12.5, 0.75) {};
		\node [style=none, doubled] (7) at (-12.5, -0.75) {};
		\node [style=Xbwdot] (8) at (-11, 0) {};
		\node [style=Xbwdot] (9) at (8.25, 0) {};
	\end{pgfonlayer}
	\begin{pgfonlayer}{edgelayer}
		\draw [style=-] (9) to (0.center);
		\draw [style=-, in=135, out=0] (6.center) to (8);
		\draw [style=-, in=-135, out=0] (7.center) to (8);
		\draw [style=-] (8) to (4.center);
	\end{pgfonlayer}
\end{tikzpicture}
\end{equation} 
It is immediate to check that these two maps form, together with their adjoints, a unital quasi-special commutative $\dagger$-Frobenius algebra, with normalisation factor $(2 \omega + 1)^n$. Furthermore, a result of Ref. \cite{Gogioso2017b} guarantees that $(\hbox{\begin{tikzpicture} [scale=1.2,transform shape] 

\def\deltax{0.3} 
\def\deltay{0.5} 


\node [dot, fill=\Dcolour] (mult) at (0,0) {};

\end{tikzpicture}}\!\!,\hbox{\begin{tikzpicture} [scale=1.2,transform shape] 

\def\deltax{0.3} 
\def\deltay{0.5} 


\node [dot, fill=\Zbwcolour] (mult) at (0,0) {};

\end{tikzpicture}}\!\!)$ is a strongly complementary pair of unital $\dagger$-Frobenius algebras.

The same proof from Section \ref{section_box} can be adapted to show that the following \textbf{position eigenstates}, indexed by $\underline{x} \in \frac{1}{\omega_{ir}}\starIntegersModPow{2\omega+1}{n}$, are the $\hbox{\begin{tikzpicture} [scale=1.2,transform shape] 

\def\deltax{0.3} 
\def\deltay{0.5} 


\node [dot, fill=\Dcolour] (mult) at (0,0) {};

\end{tikzpicture}}\!\!$-classical states (orthogonal and with square norm $(2 \omega + 1)^n$):
\begin{equation}
	\ket{\delta_{\underline{x}}} := \sum_{p_1 = -\omega_{ir}}^{+ \omega_{ir}} ... \sum_{p_n = -\omega_{ir}}^{+\omega_{ir}} \goodchi_{\underline{p}}(\underline{x})^\ast \ket{\goodchi_{\underline{p}}}
\end{equation}
Note once again the nice duality between the cut-offs of the momenta (which have magnitude bounded above by $\sqrt{n}\omega_{ir}\hbar$ and are discretised onto a lattice of infinitesimal mesh $\frac{\hbar}{\omega_{uv}}$) and the cut-offs of the positions (which have magnitude bounded above by $\sqrt{n}\omega_{uv}$ and are discretised onto a lattice of infinitesimal mesh $\frac{1}{\omega_{ir}}$). For all $\underline{x} \in \frac{1}{\omega_{ir}}\starIntegersModPow{2\omega+1}{n}$ and all standard smooth $f \in \Ltwo{\reals^{n}}$, the position eigenstates can easily be shown to satisfy the identity $\stdpart{\braket{\delta_{\underline{x}}}{f}} = f(\stdpart{\underline{x}})$: hence, the position eigenstates (as we defined them) behave exactly as expected from Dirac delta functions, and we can legitimately refer to $\hbox{\begin{tikzpicture} [scale=1.2,transform shape] 

\def\deltax{0.3} 
\def\deltay{0.5} 


\node [dot, fill=\Dcolour] (mult) at (0,0) {};

\end{tikzpicture}}\!\!$ as the \textbf{position observable} for the unbounded particle.  Furthermore, standard smooth functions span the entirety of $\Ltwo{\reals^n}$, and hence the delta functions defined above show that the subspace defined by the truncating projector $P_\realsAlgebra{n}$ spans the near-standard vectors in $\nonstd{\Ltwo{\reals^n}}$.

\paragraph {Weyl CCRs using diagrams in \texorpdfstring{$\starHilbCategory$}{Star Hilb}.}

As a sample application of our non-standard framework, we provide a diagrammatic proof of the Weyl Canonical Commutation Relations, using the $\starHilbCategory$ object $\realsAlgebra{n}$ which we have just constructed. For all $\underline{x} \in \reals^n$, let $U_{\underline{x}}$ be the unitary on $\Ltwo{\reals^n}$ corresponding to space-translation of wavefunctions by $\underline{x}$. For all $\underline{p} \in \reals^n$, let $V_{\underline{p}}$ be the unitary corresponding to momentum-boost by $\underline{p} \hbar$. Let $\underline{x}' \in \frac{1}{\omega_{ir}}\starIntegersModPow{2\omega+1}{n}$ and $\underline{p}' \in \frac{1}{\omega_{uv}}\starIntegersModPow{2\omega+1}{n}$ be such that $\stdpart{\underline{x}'} = \underline{x}$ and $\stdpart{\underline{p}'} = \underline{p}$. The unitaries $U_{\underline{x}}$ and $V_{\underline{p}}$ can be expressed diagrammatically as follows (because $\hbox{\begin{tikzpicture} [scale=1.2,transform shape] 

\def\deltax{0.3} 
\def\deltay{0.5} 


\node [dot, fill=\Zbwcolour] (mult) at (0,0) {};

\end{tikzpicture}}\!\!$ acts as space-translation on delta functions, and $\hbox{\begin{tikzpicture} [scale=1.2,transform shape] 

\def\deltax{0.3} 
\def\deltay{0.5} 


\node [dot, fill=\Dcolour] (mult) at (0,0) {};

\end{tikzpicture}}\!\!$ acts as momentum-boost on plane-waves):
\begin{equation}\label{WeylCCRUnitariesWavefunctionsRealSpace}
	\begin{tikzpicture}
	\begin{pgfonlayer}{nodelayer}
		\node [style=labelnode] (0) at (-14.25, 0) {$U_{\underline{x}}$};
		\node [style=labelnode] (1) at (3, 0) {$V_{\underline{p}}$};
		\node [style=none] (2) at (-6, 0.25) {};
		\node [style=state] (3) at (-7.5, -0.5) {$\delta_{\underline{x}'}$};
		\node [style=labelnode] (4) at (-3.75, 0.25) {};
		\node [style=labelnode] (5) at (-9, 0.25) {};
		\node [style=Zbwdot] (6) at (-6, 0.25) {};
		\node [style=labelnode] (7) at (-10.5, 0) {$\stdpartSym$};
		\node [style=none] (8) at (-3.75, 1) {};
		\node [style=none] (9) at (-3.75, -1) {};
		\node [style=none] (10) at (-9, -1) {};
		\node [style=none] (11) at (-9, 1) {};
		\node [style=labelnode] (12) at (-12.25, 0) {$=$};
		\node [style=labelnode] (13) at (8.25, 0.25) {};
		\node [style=none] (14) at (8.25, -1) {};
		\node [style=none] (15) at (13.25, 1) {};
		\node [style=none] (16) at (13.25, -1) {};
		\node [style=labelnode] (17) at (6.75, 0) {$\stdpartSym$};
		\node [style=labelnode] (18) at (13.25, 0.25) {};
		\node [style=Xbwdot] (19) at (11.25, 0.25) {};
		\node [style=labelnode] (20) at (5, 0) {$=$};
		\node [style=state] (21) at (9.75, -0.5) {$\goodchi_{\underline{p}'}$};
		\node [style=none] (22) at (8.25, 1) {};
	\end{pgfonlayer}
	\begin{pgfonlayer}{edgelayer}
		\draw [style=-, in=135, out=0] (5) to (6);
		\draw [in=-135, out=0] (3) to (6);
		\draw [style=-, in=180, out=0, looseness=0.75] (2.center) to (4);
		\draw [style=-, in=0, out=0, looseness=0.75] (8.center) to (9.center);
		\draw [style=-, in=180, out=180, looseness=0.75] (11.center) to (10.center);
		\draw [style=-, in=135, out=0] (13) to (19);
		\draw [style=-] (19) to (18);
		\draw [in=-135, out=0] (21) to (19);
		\draw [style=-, in=0, out=0, looseness=0.75] (15.center) to (16.center);
		\draw [style=-, in=180, out=180, looseness=0.75] (22.center) to (14.center);
	\end{pgfonlayer}
\end{tikzpicture}
\end{equation}
We can then deduce the Weyl Canonical Commutation relations for the position and momentum observables of a quantum particle in $n$-dimensional real space:
\begin{equation}\label{WeylCCRProofWavefunctionsRealSpace}
	\resizebox{\textwidth}{!}{\begin{tikzpicture}
	\begin{pgfonlayer}{nodelayer}
		\node [style=labelnode] (0) at (-17, 0) {$V_{\underline{p}} U_{\underline{x}}$};
		\node [style=labelnode] (1) at (20, 0) {$e^{i 2 \pi \underline{p} \cdot \underline{x}} \;U_{\underline{x}} V_{\underline{p}}$};
		\node [style=none] (2) at (-8.75, 0) {};
		\node [style=state] (3) at (8, -1) {$\goodchi_{\underline{p}'}$};
		\node [style=none] (4) at (9.5, 0) {};
		\node [style=state] (5) at (-7.25, -1) {$\goodchi_{\underline{p}'}$};
		\node [style=labelnode] (6) at (6.5, 0) {};
		\node [style=state] (7) at (-10.25, -1) {$\delta_{\underline{x}'}$};
		\node [style=state] (8) at (11, -1) {$\delta_{\underline{x}'}$};
		\node [style=labelnode] (9) at (-3.75, 0) {};
		\node [style=Zbwdot] (10) at (12.5, 0) {};
		\node [style=labelnode] (11) at (-1.25, 0) {$=$};
		\node [style=effect] (12) at (5, 0) {$\goodchi^\dagger_{-\underline{p}'}$};
		\node [style=Xbwdot] (13) at (-5.75, 0) {};
		\node [style=labelnode] (14) at (-11.75, 0) {};
		\node [style=labelnode] (15) at (14.5, 0) {};
		\node [style=state] (16) at (2.5, 0) {$\delta_{\underline{x}'}$};
		\node [style=Xbwdot] (17) at (9.5, 0) {};
		\node [style=Zbwdot] (18) at (-8.75, 0) {};
		\node [style=labelnode] (19) at (16.75, 0) {$=$};
		\node [style=labelnode] (20) at (-13.25, 0) {$\stdpartSym$};
		\node [style=none] (21) at (-3.75, 1.5) {};
		\node [style=none] (22) at (-3.75, -1.5) {};
		\node [style=none] (23) at (2, -1.5) {};
		\node [style=none] (24) at (2, 1.5) {};
		\node [style=none] (25) at (-11.75, -1.5) {};
		\node [style=none] (26) at (-11.75, 1.5) {};
		\node [style=none] (27) at (14.5, -1.5) {};
		\node [style=none] (28) at (14.5, 1.5) {};
		\node [style=labelnode] (29) at (0.5, 0) {$\stdpartSym$};
		\node [style=labelnode] (30) at (-15, 0) {$=$};
	\end{pgfonlayer}
	\begin{pgfonlayer}{edgelayer}
		\draw [style=-, in=135, out=0] (14) to (18);
		\draw [in=-135, out=0] (7) to (18);
		\draw [style=-, in=135, out=0] (2.center) to (13);
		\draw [style=-] (13) to (9);
		\draw [in=-135, out=0] (5) to (13);
		\draw [style=-, in=135, out=0] (6) to (17);
		\draw [in=-135, out=0] (3) to (17);
		\draw [style=-, in=135, out=0] (4.center) to (10);
		\draw [in=-135, out=0] (8) to (10);
		\draw [style=-] (10) to (15);
		\draw [style=-] (16) to (12);
		\draw [style=-, in=0, out=0, looseness=0.75] (21.center) to (22.center);
		\draw [style=-, in=180, out=180, looseness=0.75] (24.center) to (23.center);
		\draw [style=-, in=180, out=180, looseness=0.75] (26.center) to (25.center);
		\draw [style=-, in=0, out=0, looseness=0.75] (28.center) to (27.center);
	\end{pgfonlayer}
\end{tikzpicture}}
\end{equation}
The proof of the central equality can be carried out fully diagrammatically in the non-standard framework, using strong complementarity for $(\hbox{\begin{tikzpicture} [scale=1.2,transform shape] 

\def\deltax{0.3} 
\def\deltay{0.5} 


\node [dot, fill=\Dcolour] (mult) at (0,0) {};

\end{tikzpicture}}\!\!,\hbox{\begin{tikzpicture} [scale=1.2,transform shape] 

\def\deltax{0.3} 
\def\deltay{0.5} 


\node [dot, fill=\Zbwcolour] (mult) at (0,0) {};

\end{tikzpicture}}\!\!)$ together with the fact that delta functions are $\hbox{\begin{tikzpicture} [scale=1.2,transform shape] 

\def\deltax{0.3} 
\def\deltay{0.5} 


\node [dot, fill=\Dcolour] (mult) at (0,0) {};

\end{tikzpicture}}\!\!$-classical states and that plane-waves are $\hbox{\begin{tikzpicture} [scale=1.2,transform shape] 

\def\deltax{0.3} 
\def\deltay{0.5} 


\node [dot, fill=\Zbwcolour] (mult) at (0,0) {};

\end{tikzpicture}}\!\!$-classical states:\vspace{-3mm}
\begin{equation}\label{WeylCCRsProof1}
	\resizebox{0.8\textwidth}{!}{\begin{tikzpicture}
	\begin{pgfonlayer}{nodelayer}
		\node [style=Zbwdot] (0) at (-13, 0) {};
		\node [style=labelnode] (1) at (-16, 0) {};
		\node [style=state] (2) at (-14.5, -1) {$\delta_{\underline{x}'}$};
		\node [style=Xbwdot] (3) at (-10, 0) {};
		\node [style=none] (4) at (-13, 0) {};
		\node [style=labelnode] (5) at (-8, 0) {};
		\node [style=labelnode] (6) at (-6, 0) {$=$};
		\node [style=effect] (7) at (2.5, -1) {$\goodchi_{-\underline{p}'}^\dagger$};
		\node [style=labelnode] (8) at (-4, 0) {};
		\node [style=Zbwdot] (9) at (-1, 0) {};
		\node [style=Xbwdot] (10) at (1, 0) {};
		\node [style=state] (11) at (-2.5, -1) {$\delta_{\underline{x}'}$};
		\node [style=none] (12) at (-1, 0) {};
		\node [style=labelnode] (13) at (4, 0) {};
		\node [style=labelnode] (14) at (6, 0) {$=$};
		\node [style=Xbwdot] (15) at (11.5, -1) {};
		\node [style=Xbwdot] (16) at (11.5, 1) {};
		\node [style=labelnode] (17) at (8, 0) {};
		\node [style=state] (18) at (10, -1) {$\delta_{\underline{x}'}$};
		\node [style=Zbwdot] (19) at (13.5, 1) {};
		\node [style=Zbwdot] (20) at (13.5, -1) {};
		\node [style=labelnode] (21) at (17, 0) {};
		\node [style=effect] (22) at (15.5, -1) {$\goodchi_{-\underline{p}'}^\dagger$};
		\node [style=state] (23) at (-11.5, -1) {$\goodchi_{\underline{p}'}$};
	\end{pgfonlayer}
	\begin{pgfonlayer}{edgelayer}
		\draw [style=-, in=135, out=0] (1) to (0);
		\draw [in=-135, out=0] (2) to (0);
		\draw [style=-, in=135, out=0] (4.center) to (3);
		\draw [style=-] (3) to (5);
		\draw [style=-, in=135, out=0] (8) to (9);
		\draw [in=-135, out=0] (11) to (9);
		\draw [style=-, in=180, out=0, looseness=0.75] (12.center) to (10);
		\draw [style=-, in=180, out=60] (10) to (13);
		\draw [in=-45, out=180] (7) to (10);
		\draw [style=-, in=135, out=0] (17) to (16);
		\draw [style=-, in=180, out=45] (19) to (21);
		\draw [style=-] (18) to (15);
		\draw [style=-] (15) to (19);
		\draw [style=-] (15) to (20);
		\draw [style=-] (20) to (16);
		\draw [style=-] (16) to (19);
		\draw [style=-] (20) to (22);
		\draw [style=-, in=-135, out=0] (23) to (3);
	\end{pgfonlayer}
\end{tikzpicture}}
\end{equation}
\vspace{-3mm}
\begin{equation}\label{WeylCCRsProof2}
	\resizebox{\textwidth}{!}{\begin{tikzpicture}
	\begin{pgfonlayer}{nodelayer}
		\node [style=labelnode] (0) at (-8, 0) {$=$};
		\node [style=Xbwdot] (1) at (19, 0) {};
		\node [style=labelnode] (2) at (16, 0) {};
		\node [style=labelnode] (3) at (24, 0) {};
		\node [style=Zbwdot] (4) at (22, 0) {};
		\node [style=none] (5) at (19, 0) {};
		\node [style=state] (6) at (20.5, -1) {$\delta_{\underline{x}'}$};
		\node [style=labelnode] (7) at (10, 0) {$=$};
		\node [style=Xbwdot] (8) at (-15.5, -1) {};
		\node [style=effect] (9) at (-11.5, -1) {$\goodchi_{-\underline{p}'}^\dagger$};
		\node [style=Zbwdot] (10) at (-13.5, -1) {};
		\node [style=labelnode] (11) at (-19, 0) {};
		\node [style=state] (12) at (-17.5, -1) {$\delta_{\underline{x}'}$};
		\node [style=Zbwdot] (13) at (-13.5, 1) {};
		\node [style=Xbwdot] (14) at (-15.5, 1) {};
		\node [style=labelnode] (15) at (-10, 0) {};
		\node [style=state] (16) at (-5.75, 0) {$\delta_{\underline{x}'}$};
		\node [style=effect] (17) at (-3.25, 0) {$\goodchi_{-\underline{p}'}^\dagger$};
		\node [style=labelnode] (18) at (8, 0) {};
		\node [style=none] (19) at (0, 0) {};
		\node [style=Zbwdot] (20) at (6, 0) {};
		\node [style=state] (21) at (4.5, -1) {$\delta_{\underline{x}'}$};
		\node [style=labelnode] (22) at (-2, 0) {};
		\node [style=effect] (23) at (1.5, -1) {$\goodchi_{-\underline{p}'}^\dagger$};
		\node [style=Xbwdot] (24) at (0, 0) {};
		\node [style=state] (25) at (17.5, -1) {$\goodchi_{\underline{p}'}$};
		\node [style=effect] (26) at (14.75, 0) {$\goodchi_{-\underline{p}'}^\dagger$};
		\node [style=state] (27) at (12.25, 0) {$\delta_{\underline{x}'}$};
	\end{pgfonlayer}
	\begin{pgfonlayer}{edgelayer}
		\draw [style=-, in=135, out=0] (2) to (1);
		\draw [style=-, in=135, out=0] (5.center) to (4);
		\draw [in=-135, out=0] (6) to (4);
		\draw [style=-] (4) to (3);
		\draw [style=-, in=135, out=0] (11) to (14);
		\draw [style=-, in=180, out=45] (13) to (15);
		\draw [style=-] (12) to (8);
		\draw [style=-] (8) to (13);
		\draw [style=-] (8) to (10);
		\draw [style=-] (10) to (14);
		\draw [style=-] (14) to (13);
		\draw [style=-] (16) to (17);
		\draw [style=-, in=135, out=45] (19.center) to (20);
		\draw [in=-135, out=0] (21) to (20);
		\draw [style=-] (20) to (18);
		\draw [style=-, in=180, out=0, looseness=0.75] (22) to (24);
		\draw [in=-45, out=180] (23) to (24);
		\draw [style=-] (10) to (9);
		\draw [style=-] (27) to (26);
		\draw [style=-, in=-135, out=0] (25) to (5.center);
	\end{pgfonlayer}
\end{tikzpicture}}
\end{equation}

\newcommand{\encodingFunction}[1]{\operatorname{enc}_{#1}}
\newcommand{\decodingFunction}[1]{\operatorname{dec}_{#1}}
\section{Quantum fields on lattices}
\label{section_fieldLattice}

In Section \ref{section_reals}, we have seen that there are so many non-standard rational numbers that the standard reals can be approximated with ease by considering an infinite lattice with infinitesimal mesh. In this Section, we will see that there also are so many non-standard integers that quantum fields living on lattices, traditionally forming a non-separable Hilbert space, can be modelled in our framework.

If $V$ is a standard finite-dimensional Hilbert space with dimension greater than one, and $X$ is some countably infinite set, then it is well known that the space $V^{\otimes X}$ of $V$-valued quantum fields on $X$ is a non-separable standard Hilbert space: any orthonormal basis would have cardinality $(\dim{V})^X$, which is strictly larger than $\aleph_0$ whenever $\dim{V}>1$ and $X$ is infinite (we use the infinite direct product $V^{\otimes X}$ of von Neumann \cite{VonNeumann1939}). 
Our model begins with the following observation: if $D,m \in \naturals$ then we have $D^m \in \naturals$, and hence by Transfer Theorem if $D,\mu \in \starNaturals$ then we must have $D^\mu \in \starNaturals$. Fix $D,m \in \naturals^+$, and observe that the natural numbers between $1$ and $D^m$ can be constructively interpreted as strings of length $m$ with characters chosen in $\{1,...,D\}$. Denote the corresponding decoding/encoding functions as follows:
\begin{align}
	\decodingFunction{D,m}:\{1,...,D^m\} \rightarrow \{1,...,D\}^m \hspace{1.5cm}\encodingFunction{D,m}:\{1,...,D\}^m \rightarrow \{1,...,D^m\}
\end{align}
By Transfer Theorem, we obtain a pair of corresponding decoding/encoding functions $\decodingFunction{D,\mu}$ and $\encodingFunction{D,\mu}$ for each pair of positive non-standard naturals $D,\mu \in \starNaturals^+$; we will use underlined letters $\underline{s}$ to denote strings seen as functions $\underline{s} \in \{1,...,D\}^\mu$, and undecorated letters $s$ to denote the corresponding encodings of strings as numbers $s \in \{1,...,D^\mu\}$. 

Take some orthonormal family $\ket{e_{\underline{s}}} \in |\SpaceH|$ of vectors in some non-standard Hilbert space $|\SpaceH|$, and a family $\psi_{\underline{s}} \in \starComplexs$ of non-standard complex numbers, both indexed by the internal set of all strings $\underline{s} \in \{1,...,D\}^\mu$ for some positive non-standard naturals $D,\mu \in \starNaturals^+$. By using the decoding function, we can always construct a vector $\ket{\psi}$ of $|\SpaceH|$ as follows:
\begin{equation}
\ket{\psi} := \sum_{s=1}^{D^\mu} \psi_{\decodingFunction{D,\mu}(s)} \ket{e_{_{\decodingFunction{D,\mu}(s)}}}
\end{equation}
Now consider an object $\SpaceH:=(|\SpaceH|,P_\SpaceH)$ of $\starHilbCategory$, with $P_\SpaceH := \sum_{d=1}^{\dim{\SpaceH}} \ket{e_d}\bra{e_d}$ for some orthonormal family $\ket{e_d}_{d=1}^{\dim{\SpaceH}}$: we wish to construct an object $\SpaceH^{\otimes \starIntegersModPow{2\omega+1}{n}}$ corresponding to a $\SpaceH$-valued quantum field living on the lattice $\starIntegersModPow{2\omega+1}{n}$. Define the shorthands $D := \dim{\SpaceH}$ and $\mu := (2 \omega + 1)^n$, and consider the following orthonormal family $\ket{e_{\underline{s}}}$ of non-standard states in $|\SpaceH|^{\otimes \mu}$ (a non-standard Hilbert space which exists by Transfer Theorem) indexed by strings $\underline{s} \in \{1,...,D\}^{\mu}$:
\begin{equation}\label{eqn_latticeFieldBasis}
\ket{e_{\underline{s}}} := \bigotimes_{k_1 = -\omega}^{+\omega} ... \bigotimes_{k_n=-\omega}^{+ \omega} \ket{e_{\underline{s}(\underline{k})}}
\end{equation}
We introduced the following shorthand to access the characters of the indexing strings:
\begin{equation}
\underline{s}(\underline{k}) :\equiv  \;\;\underline{s}\Big( \encodingFunction{2\omega+1,n}\big(k_1+\omega+1,...,k_n+\omega+1\big) \Big) \in \{1,...,D\} 
\end{equation}
We have chosen to introduce the shorthand above because $\underline{s}$ is technically $\underline{s}: \{1,...,\mu\} \rightarrow \{1,...,D\}$, but it is more physically significant to treat it as $\underline{s}:\{-\omega,...,+\omega\}^n \rightarrow \{1,...,D\}$, since we are working in the context of $n$-dimensional lattices.

Traditionally, this procedure would result in an uncountable family, which cannot be adequately summed in the standard framework. In our non-standard framework, on the other hand, we can use the encoding/decoding trick to sum it and define a legitimate truncating projector:
\begin{equation}\label{truncatingProjectorLatticeField}
P_{ \SpaceH^{\otimes\starIntegersModPow{2\omega+1}{n}}} := \sum_{s = 1}^{D^\mu} \ket{e_{\decodingFunction{D,\mu}(s)}}\bra{e_{\decodingFunction{D,\mu}(s)}} 
\end{equation}
We take the corresponding object $\SpaceH^{\otimes\starIntegersModPow{2\omega+1}{n}} := (|\SpaceH|^{\otimes\mu}, P_{ \SpaceH^{\otimes\starIntegersModPow{2\omega+1}{n}}})$ of $\starHilbCategory$ to model $\SpaceH$-valued quantum fields on the lattice $\starIntegersModPow{2\omega+1}{n}$ within our framework. As a matter of convenience, we will adopt a more slender notation for sums over strings, leaving the decoding step $s \mapsto \underline{s} := \decodingFunction{D,\mu}(s)$ implicit:
\begin{equation}
P_{ \SpaceH^{\otimes\starIntegersModPow{2\omega+1}{n}}} = \sum_{\underline{s}} \ket{e_{\underline{s}}}\bra{e_{\underline{s}}} 
\end{equation} 

Now consider a decomposition $P_\SpaceH := \sum_{d=1}^{D}\ket{\phi_{d}^{(\underline{k})}}\bra{\phi_{d}^{(\underline{k})}}$ of the truncating projector $P_\SpaceH$ for the quantum system $\SpaceH$ in terms of orthonormal families specified at each point $\underline{k}$ of the lattice, and let  $\big(\hbox{\begin{tikzpicture} [scale=1.2,transform shape] 

\def\deltax{0.3} 
\def\deltay{0.5} 


\node [dot, fill=\Zbwcolour] (mult) at (0,0) {};

\end{tikzpicture}}\!\!^{(\underline{k})}\big)_{\underline{k} \in \starIntegersModPow{2\omega+1}{n}}\!$ be the associated non-degenerate observable on $\SpaceH$. Then the truncating projector $ P_{ \SpaceH^{\otimes\starIntegersModPow{2\omega+1}{n}}}$ for the quantum field is correspondingly decomposed as $P_{ \SpaceH^{\otimes\starIntegersModPow{2\omega+1}{n}}} := \sum_{\underline{s}} \ket{\phi_{\underline{s}}}\bra{\phi_{\underline{s}}}$, where $\ket{\phi_{\underline{s}}} := \bigotimes_{k_1 = -\omega}^{+\omega} ... \bigotimes_{k_n=-\omega}^{+ \omega} \ket{\phi^{(\underline{k})}_{\underline{s}(\underline{k})}}$. The associated non-degenerate observable for the quantum field is given by the following unital special commutative $\dagger$-Frobenius algebra:
\begin{equation}\label{fieldLatticeFrobAlgebra}
\begin{tikzpicture}
	\begin{pgfonlayer}{nodelayer}
		\node [style=none, doubled] (0) at (3.5, 0) {};
		\node [style=none, doubled] (1) at (-9.5, 0) {$:=$};
		\node [style=tightlabelnode] (2) at (10,-0.25) {$\sum\limits_{\underline{s}}  \bra{\phi_{\underline{s}}}$};
		\node [style=none, doubled] (3) at (7.5, 0) {$:=$};
		\node [style=none] (4) at (-14, 0) {};
		\node [style=tightlabelnode] (5) at (-4.75, -0.25) {$\sum\limits_{\underline{s}} \ket{\phi_{\underline{s}}} \otimes \ket{\phi_{\underline{s}}} \otimes \bra{\phi_{\underline{s}}}$};
		\node [style=none, doubled] (6) at (-11, 0.75) {};
		\node [style=none, doubled] (7) at (-11, -0.75) {};
		\node [style=Zbwdot] (8) at (-12.5, 0) {};
		\node [style=Zbwdot] (9) at (5.25, 0) {};
	\end{pgfonlayer}
	\begin{pgfonlayer}{edgelayer}
		\draw [style=-] (9) to (0.center);
		\draw [style=-] (4.center) to (8);
		\draw [style=-, in=180, out=45] (8) to (6.center);
		\draw [style=-, in=180, out=-45] (8) to (7.center);
	\end{pgfonlayer}
\end{tikzpicture}
\end{equation}

In Sections \ref{section_box} and \ref{section_lattice}, we constructed objects $\torusAlgebra{n}$ and $\integersAlgebra{n}$ having $\nonstd{\Ltwo{\torusGroup{n}}}$ and $\nonstd{\Ltwo{\integers^n}}$ as their underlying Hilbert spaces, and we considered truncating projectors obtained from the extension of a standard orthonormal basis. In Section \ref{section_reals}, we constructed an object $\realsAlgebra{n}$ having $\nonstd{\Ltwo{\reals^{n}}}$ as its underlying Hilbert space, but we constructed the truncating projector using a genuinely non-standard orthonormal basis. In all three cases, the underlying Hilbert space is the non-standard extension $\nonstd{V}$ of a separable standard Hilbert space $V$, and the connection to the traditional quantum mechanical formalism is guaranteed by Theorem~\ref{thm_FundamentalTheorem}. Unfortunately, Theorem \ref{thm_FundamentalTheorem} is not applicable in this Section: in order to establish a connection between $\SpaceH^{\otimes\starIntegersModPow{2\omega+1}{n}}$ and the traditional model for quantum fields on lattices, we will formulate a suitable universal property for $\SpaceH^{\otimes\starIntegersModPow{2\omega+1}{n}}$. In the remainder of this Section, we will assume that $|\SpaceH| = \nonstd{V}$ for some standard separable Hilbert space $V$.

Consider the infinite direct sum of standard Hilbert spaces $\prod_{\underline{k} \in \integers^n} V$. If $W$ is another standard Hilbert space, we say that a function $\tilde{f}: \big(\prod_{\underline{k} \in \integers^n} V\big) \rightarrow W$ is \textbf{multilinear} if for each $\underline{k} \in \integers^n$ and each ${u}: \integers^n\backslash\{\underline{k}\} \rightarrow V$ the function $\restrict{\tilde{f}}{\underline{k},{u}}: V \rightarrow W$ defined by $\restrict{\tilde{f}}{\underline{k},{u}} (v) := \tilde{f}({u}\cup\{\underline{k} \mapsto v\})$ is linear. Now consider the non-standard extension $\prod_{\underline{k} \in \starIntegers^n}\!\! \nonstd{V}$ of the direct sum, and let $\prod_{\underline{k} \in \starIntegersModPow{2\omega+1}{n}}\!\! \SpaceH$ be the object of $\starHilbCategory$ obtained by restricting non-standard internal maps $\nonstd{\varphi}:\starIntegers^n\! \rightarrow\!\! \nonstd{V}$ to non-standard internal maps $\nonstd{\varphi}:\starIntegersModPow{2\omega+1}{n}\! \rightarrow\!\! \nonstd{V}$, extending the projector $P_\SpaceH$ to the direct sum $\prod_{\underline{k} \in \starIntegersModPow{2\omega+1}{n}}\!\!\!\!\!\! \nonstd{V}$ by pointwise action. Multilinear maps $\tilde{F}: \big(\prod_{\underline{k} \in \starIntegersModPow{2\omega+1}{n}}\!\! \SpaceH\big) \rightarrow \SpaceK$ are defined analogously to the standard case.

\begin{theorem}[\textbf{Universal property for the infinite tensor product $\SpaceH^{\otimes\starIntegersModPow{2\omega+1}{n}}$}]\label{thm_latticeFieldsUniversalProperty}\hfill\\
Let $|\SpaceH| = \nonstd{V}$, and define $\theta: \big(\prod_{\underline{k} \in \starIntegersModPow{2\omega+1}{n}} \SpaceH \big) \rightarrow \SpaceH^{\otimes\starIntegersModPow{2\omega+1}{n}}$ to be the following multilinear map:
\begin{equation}
\theta \big(\,\underline{k} \mapsto \sum_{d=1}^D v_d(\underline{k})\, \ket{e_d} \,\big) := \sum_{\underline{s}} \Big(\; \bigotimes_{k_1 = -\omega}^{+\omega} ... \bigotimes_{k_n=-\omega}^{+ \omega}  v_{\underline{s}(\underline{k})}(\underline{k})\, \ket{e_{\underline{s}(\underline{k})}} \;\Big)
\end{equation}
Then for every object $\SpaceK$ with $|\SpaceK| = \nonstd{W}$ and every multilinear map $\tilde{F}: \big(\prod_{\underline{k} \in \starIntegersModPow{2\omega+1}{n}}\SpaceH\big) \rightarrow \SpaceK$ there is a unique linear map $F:  \SpaceH^{\otimes\starIntegersModPow{2\omega+1}{n}} \rightarrow \SpaceK$ such that the following diagram commutes:
\begin{equation}\label{latticeFieldsUniversalProperty}
\begin{tikzpicture}
	\begin{pgfonlayer}{nodelayer}
		\node [style=tightlabelnode] (0) at (-3.5, 3) {$\prod\limits_{\underline{k} \in \starIntegersModPow{2\omega+1}{n}}\!\!\!\!\SpaceH$};
		\node [style=labelnode] (1) at (3, 3) {$\SpaceH^{\otimes\starIntegersModPow{2 \omega + 1}{n}}$};
		\node [style=labelnode] (2) at (3, -1) {$\SpaceK$};
		\node [style=none] (3) at (0, 3.5) {$\theta$};
		\node [style=none] (4) at (3.75, 1) {$F$};
		\node [style=none] (5) at (-0.5, 0) {$\tilde{F}$};
	\end{pgfonlayer}
	\begin{pgfonlayer}{edgelayer}
		\draw [style=->] (0) to (1);
		\draw [style=->] (1) to (2);
		\draw [style=->] (0) to (2);
	\end{pgfonlayer}
\end{tikzpicture}
\end{equation}
\end{theorem}
\begin{proof}
The map $F$ can be defined on the standard orthonormal basis of $\SpaceH^{\otimes\starIntegersModPow{2\omega+1}{n}}$ as follows: 
\begin{equation}
F \Big( \bigotimes_{k_1 = -\omega}^{+\omega} ... \bigotimes_{k_n=-\omega}^{+ \omega}\ket{e_{\underline{s}(\underline{k})}} \Big) := \tilde{F}  \Big( \underline{k} \mapsto \ket{e_{\underline{s}(\underline{k})}} \Big)
\end{equation}
Commutativity of Diagram \ref{latticeFieldsUniversalProperty} and uniqueness of $F$ are both consequences of the following equation for $\tilde{F}$, together with the observation that $\theta \, \big( \underline{k} \mapsto \ket{e_{\underline{s}(\underline{k})}} \big) = \bigotimes_{k_1 = -\omega}^{+\omega} ... \bigotimes_{k_n=-\omega}^{+ \omega}\ket{e_{\underline{s}(\underline{k})}}$:
\begin{equation}\label{nonstandardMultilinearity}
\tilde{F}\Big(\,\underline{k} \mapsto \sum_{d=1}^D v_d(\underline{k}) \,\ket{e_d} \,\Big)\;\; = \;\; \sum_{\underline{s}}\Big( \tilde{F}\big(\, \underline{k} \mapsto \ket{e_{\underline{s}(\underline{k})}} \,\big) \!\!\!\prod_{\underline{k} \in \starIntegersModPow{2\omega+1}{n}}\!\!v_{\underline{s}(\underline{k})}(\underline{k}) \Big)
\end{equation}
Equation \ref{nonstandardMultilinearity} itself is a consequence of Transfer Theorem, because the corresponding standard statement is valid for all standard multilinear $\tilde{f}: \big(\prod_{\underline{k} \in \{-m,...,+m\}^n} V\big) \rightarrow W$, for any choice of $m \in \naturals$.
\end{proof}

\section{Quantum fields in real space}
\label{section_QFT}
In Section \ref{section_reals}, we have seen that there are so many non-standard rational numbers that the standard reals can be approximated with ease by considering an infinite lattice with infinitesimal mesh. In Section \ref{section_fieldLattice}, we have seen that there are so many non-standard integers that quantum fields on lattices can be modelled in our framework. In this Section, we will put both tricks together to show that, in fact, quantum fields on unbounded real space can also be modelled in our framework (a similar argument applies to tori).

Consider again an object $\SpaceH:=(|\SpaceH|,P_\SpaceH)$ of $\starHilbCategory$, with $P_\SpaceH := \sum_{d=1}^{D} \ket{e_d}\bra{e_d}$ for some orthonormal family $\ket{e_d}_{d=1}^{D}$. We fix two odd infinite natural numbers $\omega_{uv},\omega_{ir} \in \starNaturals$, and let $2\omega+1:= \omega_{uv} \omega_{ir}$. We wish to construct the system $\SpaceH^{\otimes \frac{1}{\omega_{uv}}\starIntegersModPow{2\omega+1}{n}}$ corresponding to a $\SpaceH$-valued quantum field living on the lattice $\frac{1}{\omega_{uv}}\starIntegersModPow{2\omega+1}{n}$ in $\starReals^n$, which we have seen in Section \ref{section_reals} to approximate the real space $\reals^n$ to infinitesimal mesh~$\frac{1}{\omega_{uv}}$ (and all the way up to some infinity $\omega_{ir}$, where the lattice \inlineQuote{circles around}). 

As objects of $\starHilbCategory$, the space $\SpaceH^{\otimes\starIntegersModPow{2\omega+1}{n}}$ we constructed in the previous Section and the space $\SpaceH^{\otimes \frac{1}{\omega_{uv}}\starIntegersModPow{2\omega+1}{n}}$ we wish to construct in this Section are isomorphic: to see this, it suffices to consider the isomorphism of abelian groups sending $\underline{k} \in \starIntegersModPow{2\omega+1}{n}$ to $\underline{p} := \frac{1}{\omega_{uv}} \underline{k} \in \frac{1}{\omega_{uv}}\starIntegersModPow{2\omega+1}{n}$. As a consequence, the only formal distinction is that we will write our summations over $\underline{p}$ rather than $\underline{k}$, i.e. we will use the same shorthand $\sum_{p=-\omega_{ir}}^{+\omega_{ir}} f(p) :\equiv \sum_{k=-\omega}^{+\omega} f(k/\omega_{uv})$ that we introduced in Section \ref{section_reals}.

When it comes to quantum field theory, however, the two objects $\SpaceH^{\otimes\starIntegersModPow{2\omega+1}{n}}$ and $\SpaceH^{\otimes \frac{1}{\omega_{uv}}\starIntegersModPow{2\omega+1}{n}}$ have very different interpretations, corresponding to the different ways in which the underlying lattice is immersed into $n$-dimensional non-standard real space $\nonstd{\reals^n}$ (and successively related to standard real space~$\reals^n$). In the previous Section, the underlying lattice was embedded as the lattice $\starIntegersModPow{2\omega+1}{n}$ of standard mesh $1$ in $\nonstd{\reals^n}$, while in this section the underlying lattice is embedded as the lattice $\frac{1}{\omega_{uv}}\starIntegersModPow{2\omega+1}{n}$of infinitesimal mesh $\frac{1}{\omega_{uv}}$. From the perspective of non-standard reals, the two lattices are equivalent (rescaling by a factor of $\omega_{uv}\in\starReals$), but from the perspective of standard reals they are extremely different: restricting to near-standard reals and quotienting by infinitesimal equivalence sends $\starIntegersModPow{2\omega+1}{n}$ to the lattice $\integers^n$, while $\frac{1}{\omega_{uv}}\starIntegersModPow{2\omega+1}{n}$ covers the entirety of $\reals^n$. We now provide justification for the interpretation of  $\SpaceH^{\otimes\frac{1}{\omega_{uv}}\starIntegersModPow{2\omega+1}{n}}$ as a space of quantum fields on real space.

Just as we did in the previous Section, we will assume that $|\SpaceH| = \nonstd{V}$ for separable standard $V$. Consider the direct integral of Hilbert spaces $\int_{\reals^n}^{\oplus} \,V \,d\underline{p}$ (in the sense of von Neumann \cite{Neumann1949}), i.e. the space of square-integrable functions $\varphi: \reals^n \rightarrow V$, together with the inner product $\braket{\psi}{\varphi} := \int_{\reals^n} \,\braket{\psi(\underline{p})}{\varphi(\underline{p})} \, d \underline{p}$. Consider its non-standard extension, and let $\prod_{\underline{p} \in \frac{1}{\omega_{uv}}\starIntegersModPow{2\omega+1}{n}}\!\!\!\!\! \SpaceH$ be the object obtained by restricting non-standard internal maps $\nonstd{\varphi}: \starReals^n \rightarrow V$ to non-standard internal maps $\nonstd{\varphi}: \frac{1}{\omega_{uv}}\starIntegersModPow{2\omega+1}{n}  \rightarrow V$, and appropriately extending the projector $P_\SpaceH$ to the direct sum $\prod_{\underline{p} \in \frac{1}{\omega_{uv}}\starIntegersModPow{2\omega+1}{n}}\!\!\!\!\!\!\nonstd{V}$. 

The universal property for the infinite tensor product $\SpaceH^{\otimes\frac{1}{\omega_{uv}}\starIntegersModPow{2\omega+1}{n}}$ takes a form similar to that of Theorem \ref{thm_latticeFieldsUniversalProperty}, and has a similar proof: one only needs to use $\prod_{\underline{p} \in \frac{1}{\omega_{uv}}\starIntegersModPow{2\omega+1}{n}} \SpaceH$ in place of the original $\prod_{\underline{k} \in \starIntegersModPow{2\omega+1}{n}} \SpaceH$. The non-trivial part is the connection between $\prod_{\underline{p} \in \frac{1}{\omega_{uv}}\starIntegersModPow{2\omega+1}{n}} \SpaceH$ and the direct integral $\int_{\reals^n}^{\oplus} \,V \,d\underline{p}$. In the previous Section, the connection between $\prod_{\underline{k} \in \starIntegersModPow{2\omega+1}{n}} \SpaceH$ and the direct sum  $\prod_{\underline{k} \in \integers^n} V$ was immediate: a function $\varphi: \integers^n \rightarrow V$ extends to a function $\nonstd{\varphi}: \starIntegers^n \rightarrow \nonstd{V}$, which then restricts to $\nonstd{\varphi}: \starIntegersModPow{2\omega+1}{n} \rightarrow \nonstd{V}$ and then back to $\varphi: \integers^n \rightarrow V$. In this Section, we note that $\int_{\reals^n}^{\oplus} \,V \,d\underline{p}$ is the $L^2$ completion of the space of continuous square-integrable functions $\varphi: \reals^n \rightarrow V$, so it is enough to show that these can be reconstructed from the corresponding $\nonstd{\varphi}: \frac{1}{\omega_{uv}}\starIntegersModPow{2\omega+1}{n} \rightarrow \nonstd{V}$. But this is indeed the case: if $\varphi$ is continuous then we have that $\stdpart{x} = \stdpart{y}$ implies $\stdpart{\nonstd{\varphi}(x)} = \stdpart{\nonstd{\varphi}(y)}$, and in particular we can obtain $\varphi$ back from $\nonstd{\varphi}$ in a unique way by setting $\varphi(z) := \stdpart{\nonstd{\varphi}(x)}$ for any $x \in \frac{1}{\omega_{uv}}\starIntegersModPow{2\omega+1}{n}$ such that $\stdpart{x} = z$.

\section{Towards Quantum Field Theory in Categorical Quantum Mechanics}
\label{section_QFTexplained}

A legitimate question to ask at this point is: How does Quantum Field Theory fit into the framework we described? Why are we talking about \inlineQuote{quantum fields} in the context of certain infinite tensor products? 

As part of canonical quantisation, classical fields from the Lagrangian formalism are translated into certain operator-valued distributions, also known as \textit{field operators}, acting upon quantum states living in a Fock space. Using the field operators, the classical Lagrangian can be translated into the dynamics and interactions of the quantum field theory, so it is no surprise that they occupy the vast majority of the literature dedicated to the subject. 

It is worth noting, however, that the field operators play a very different role from the classical fields that they originally quantised: classical fields \textit{are} states of a classical system, while field operators \textit{act upon} states of a quantum system (e.g. the vacuum). In this sense, the closest correspondents in quantum field theory to the fields of classical field theory or the wavefunctions of quantum mechanics are, in fact, the quantum states in the Fock space. Just as $\complexs^2$ is the space of quantum states for a qubit, so the Fock space is the space of quantum states for a quantum field. And just as we freely refer to the object $\complexs^2$ as a qubit, so we take the liberty to refer to the Fock space as a \textbf{quantum field}. We will use the term \textbf{field operator} when talking about the operator-valued distributions obtained by canonical quantisation.

Let's consider the textbook example of the real scalar field, a relativistic classical field $\phi(\underline{x},t)$ satisfying the \textbf{Klein-Gordon equation}:
\begin{equation}
	\partial_\mu \partial^\mu \phi + m^2 \phi = 0
\end{equation}
When looking at the field in momentum space $\phi(\underline{p},t)$, the Klein-Gordon equation becomes:
\begin{equation}
	\Big(\frac{\partial^2}{\partial t} + (|p|^2+m^2)\Big) \phi(\underline{p},t) = 0
\end{equation}
Hence a momentum space solution $\phi(\underline{p},t)$ to the Klein-Gordon equation can be thought of as a field of simple harmonic oscillators, each oscillator vibrating with its own amplitude and at a frequency given by $\nu_{\underline{p}} := \sqrt{|\underline{p}|^2+m^2}$ for each point $\underline{p} \in \reals^3$ of momentum space. In order to quantise the real scalar field $\phi$, we simply need to quantise the simple harmonic oscillators. We do so in our non-standard framework. 

Consider the object $\SpaceH$ of $\starHilbCategory$ defined as follows, where $\tau$ is some infinite non-standard natural and $\ket{n}_{n \in \starNaturals}$ is the standard orthonormal basis for $\nonstd{\Ltwo{\naturals}}$: 
\begin{equation}
	\SpaceH := \Big(\nonstd{\Ltwo{\naturals}}, \sum_{n=0}^{\tau} \ket{n}\bra{n}\Big)
\end{equation}
We will think of $\SpaceH$ as the non-standard counterpart for a quantum harmonic oscillator: the states $\ket{n}$ correspond to energy eigenstates for the oscillator, and we extended our range of energy values all the way up to some infinite natural $\tau$. We define the \textbf{ladder operators} $a$ and $a^\dagger$ on $\SpaceH$ as follows:
\begin{equation}
	a \ket{n} =
		\begin{cases}
			0 &\text{ if } n = 0 \\
			\sqrt{n} \ket{n-1} &\text{ otherwise}
		\end{cases}
	\hspace{3cm} 
	a^\dagger \ket{n} = 
		\begin{cases}
			0 &\text{ if } n = \tau \\
			\sqrt{n+1} \ket{n+1} &\text{ otherwise}
		\end{cases}
\end{equation}
It is easy to check that these operators satisfy the usual canonical commutation relations, up to a correction factor accounting for the truncation of energy above the infinite $\tau$:
\begin{equation}
	[a,a^\dagger] = \id{\SpaceH} - (\tau+1)\ket{\tau}\bra{\tau}
\end{equation}
When restricting ourselves to finite energy states, these operators are exactly the ladder operators for the quantum harmonic oscillator. We then proceed to define the \textbf{number operator} $N:=a^\dagger a$, and we obtain the usual property and commutators for it (no correction this time):
\begin{equation}
	N \ket{n} = n \ket{n} \hspace{3cm} [N,a^\dagger] = a^\dagger \hspace{3cm} [N,a] = -a
\end{equation}
The number operator is associated to a $\dagger$-SCFA on $\SpaceH$, the \textbf{number observable}, with $\ket{n}$ as classical states. For a quantum harmonic oscillator of frequency $\nu$, the Hamiltonian can finally be defined as:
\begin{equation}
	H := \hbar \nu N
\end{equation}
Aside perhaps for the correction term in the canonical commutation relation, this is exactly what we would expect the non-standard version of the quantum harmonic oscillator to look like, and the traditional quantum harmonic oscillator is recovered exactly by restricting to states of finite energy.

We saw before that a solution to the Klein-Gordon can be interpreted to describe a field of simple harmonic oscillators at each point $\underline{p} \in \reals^3$ of momentum space, vibrating independently with frequencies given by the expression $\nu_{\underline{p}} = \sqrt{|\underline{p}|^2 + m^2}$. The natural quantisation of such a scenario involves considering independent quantum harmonic oscillators at each point $\underline{p} \in \reals^3$ of momentum space, i.e. an infinite direct product of separable Hilbert spaces over the 3-dimensional continuum. Because such a space would be mathematically unwieldy, and because only finite energy states are deemed to be physically interesting, the infinite direct product of quantum harmonic oscillators is never constructed, and the Fock space is considered instead. The Fock space is the Hilbert space of joint states for the quantum harmonic oscillators which is spanned by those separable states involving only finitely many oscillators not in their ground state: the state $\ket{n}$ for the oscillator at point $\underline{p} \in \reals^3$ is considered to count the number of quantum particles with definite momentum $\underline{p}$, and the Fock space is spanned by all states containing finitely many particles.

Within our non-standard framework, we don't have to worry about infinite tensor products, and we don't have to restrict ourselves to finite energy states or finite number of particles: as a consequence, we quantise the real scalar field $\phi$ by constructing the field of quantum harmonic oscillators in all its glory. This can be done by considering the space  $\SpaceH^{\otimes \frac{1}{\omega_{uv}}\starIntegersModPow{2\omega+1}{3}}$ defined in Section \ref{section_QFT} above: we discretise momentum space to an infinite lattice $\frac{1}{\omega_{uv}}\starIntegersModPow{2\omega+1}{3}$ of infinitesimal mesh $1/\omega_{uv}$, and we place an independent quantum harmonic oscillator $\SpaceH$ at each point of the lattice (with varying frequency $\nu_{\underline{p}}$). 

For each $\underline{p} \in  \frac{1}{\omega_{uv}}\starIntegersModPow{2\omega+1}{3}$, we write $a_{\underline{p}}$ and $a^\dagger_{\underline{p}}$ for the ladder operators acting on the quantum Harmonic oscillator at $\underline{p}$ (tensored with the identity on all other oscillators), and $\ket{n,\underline{p}}_{n=0}^\tau$ for the orthonormal basis of the oscillator at $\underline{p}$. We define the rescaled versions $a(\underline{p}):=\sqrt{\omega_{uv}^3}a_{\underline{p}}$ and $a^\dagger(\underline{p}):=\sqrt{\omega_{uv}^3}a^\dagger_{\underline{p}}$, which satisfy the commutation relations $[a(\underline{p}),a(\underline{q})] = [a^\dagger(\underline{p}),a^\dagger(\underline{q})] = 0$ and:
\begin{equation}
	[a(\underline{p}),a^\dagger(\underline{q})] = 
	\begin{cases}
		\omega_{uv}^3 \big(\id{}-(\nu+1)\ket{\nu,\underline{p}}\bra{\nu,\underline{p}}\big) & \text{ if } \underline{p}=\underline{q}\\
		0 & \text{ otherwise}
	\end{cases}
\end{equation}
The usual field operators $\pi(\underline{x})$ and $\phi(\underline{x})$ can be defined from $a(\underline{p})$ and $a^\dagger(\underline{p})$ through the following discretised integral, for all points $\underline{x} \in \frac{1}{\omega_{ir}}\starIntegersModPow{2\omega+1}{3}$ in space:
\begin{align}
	\phi(\underline{x}) &:= \sum_{\underline{p}} \frac{1}{\omega_{uv}^3} \frac{1}{\sqrt{\nu_{\underline{p}}}} \Big[ a(\underline{p})e^{i2\pi \underline{p}\cdot \underline{x}} + a^\dagger(\underline{p})e^{-i2\pi \underline{p}\cdot \underline{x}} \Big] \nonumber \\
	\pi(\underline{x}) &:= \sum_{\underline{p}} \frac{1}{\omega_{uv}^3} (-i)\frac{\sqrt{\nu_{\underline{p}}}}{2} \Big[ a(\underline{p})e^{i2\pi \underline{p}\cdot \underline{x}} - a^\dagger(\underline{p})e^{-i2\pi \underline{p}\cdot \underline{x}} \Big] 
\end{align}
The field operators satisfy commutation relations similar to the ones of the rescaled ladder operators, as would be expected. For every $\textbf{n}:\frac{1}{\omega_{uv}}\starIntegersModPow{2\omega+1}{3} \rightarrow \{0,...,\tau\}$, we can define the state $\ket{\textbf{n}} := \otimes_{\underline{p}} \ket{\textbf{n}(\underline{p}),\underline{p}}$. Then the discretised integral of the (rescaled) number observables for all quantum harmonic oscillators at all points $\underline{p} \in \frac{1}{\omega_{uv}}\starIntegersModPow{2\omega+1}{3}$ of momentum space gives rise to the \textbf{number operator} $N$ on $\SpaceH^{\otimes \frac{1}{\omega_{uv}}\starIntegersModPow{2\omega+1}{3}}$:
\begin{equation}
	N  := \sum_{\underline{p}} \frac{1}{\omega_{uv}^3} a^\dagger({\underline{p}}) a({\underline{p}}) = \sum_{\underline{p}} a^\dagger_{\underline{p}}a_{\underline{p}} =\sum_{\textbf{n}}  \Big(\sum_{\underline{p}}\textbf{n}(\underline{p})\Big) \ket{\textbf{n}}\bra{\textbf{n}}
\end{equation}
The \textbf{Hamiltonian} for the quantum field is similarly obtained as a discretised integral:
\begin{equation}
	H :=  \sum_{\underline{p}} \frac{1}{\omega_{uv}^3} \hbar  \nu_{\underline{p}} a^\dagger({\underline{p}}) a({\underline{p}}) = \sum_{\underline{p}} \hbar \nu_{\underline{p}} a^\dagger_{\underline{p}} a_{\underline{p}} =\sum_{\textbf{n}}  \Big(\sum_{\underline{p}} \hbar \nu_{\underline{p}}\textbf{n}(\underline{p})\Big) \ket{\textbf{n}}\bra{\textbf{n}}
\end{equation}
The traditional Fock space is recovered by considering the states $\ket{\textbf{n}}$ with finite energy $\bra{\textbf{n}} H \ket{\textbf{n}}$ (i.e. those with a finite number of particles, all having finite momenta). The corresponding number of particles at a standard point $\underline{q} \in \reals^3$ of standard momentum space, which we will denote by $\stdpart{\textbf{n}}(\underline{q})$, is then given by the following expression:
\begin{equation}
	\stdpart{\textbf{n}}(\underline{q}) := \hspace{-3mm} \sum_{\substack{\underline{p} \in  \frac{1}{\omega_{uv}}\starIntegersModPow{2\omega+1}{3} \\\text{such that } \stdpart{\underline{p}}= \underline{q}}} \hspace{-3mm} \textbf{n}(\underline{p})	
\end{equation}
The further development of traditional QFT machinery within our framework is left to future work.

\newpage
\section{Conclusions and Future Work}
\label{section_conclusions}

\vspace{-2mm}

In the first section of this work, we have presented a more mature formulation of the category $\starHilbCategory$, refining and expanding the original definition from Ref. \cite{Gogioso2016b} in a number of ways. Firstly, the new definition is no longer restricted to standard Hilbert spaces and non-standard extensions of standard orthonormal bases, but instead allows all kinds of non-standard Hilbert spaces and orthonormal families. Secondly, the new definition is basis-independent and has a neater categorical presentation as a full sub-category of the Karoubi envelope for the category of non-standard Hilbert spaces. Thirdly, objects are no longer self-dual, and compact closure is now formulated in a basis-independent way, analogous to the one used in $\fdHilbCategory$. Backward compatibility is guaranteed by the fact that the category $\starHilbCategory$ originally defined in Ref. \cite{Gogioso2016b} is equivalent to a full-subcategory of the category $\starHilbCategory$ redefined in this work.

Our new definition allowed us to push the framework beyond its original limitations, and the bulk of this work was dedicated to the explicit constructions of five families of infinite-dimensional quantum systems that are of interest to the practising quantum theorist. In Sections \ref{section_box} and \ref{section_lattice} we have presented the quantum systems for particles in boxes with periodic boundary conditions and particles on lattices: both constructions were already within reach of the original definition of $\starHilbCategory$, and the special case of a particle in a one-dimensional box with periodic boundary conditions was already explored in Ref. \cite{Gogioso2016b}. The first real application of our extended $\starHilbCategory$ category has come in Section \ref{section_reals}, where it was put to work in presenting the quantum system for particles in $\reals^n$. Key to this construction have been the use of a truly non-standard orthonormal basis (i.e. not the extension of a standard one), together with an approximation of $\reals^n$ achieved by using a non-standard lattice of infinitesimal mesh in $\starReals^n$.

We have also seen that our extended definition allows for the treatment of certain cases of interest in quantum field theory. Thanks to a key observation about exponentials of infinite natural numbers---which are themselves infinite natural numbers by Transfer Theorem---and exploiting the freedom to work with non-separable spaces, we have constructed in Section \ref{section_fieldLattice} a quantum system suitable for the treatment of quantum fields on a cubic lattice $\integers^n$. Finally, in Section \ref{section_QFT} we have combined the ideas of Sections \ref{section_reals} and~\ref{section_fieldLattice} to construct a quantum system suitable for the treatment of quantum fields in $\reals^n$, and in Section \ref{section_QFTexplained} we have provided a first direct link to the traditional quantum field theoretic framework.

This work is a significant development of the original Ref. \cite{Gogioso2016b}, and provides a solid basis for the application of algebraic and diagrammatic methods from CQM to infinite-dimensional quantum mechanics and quantum field theory. From here, we foresee a number of interesting further developments and applications, some of which are briefly detailed below.

\vspace{-3mm}

\paragraph{Future work.} To begin with, an extension of Theorem \ref{thm_FundamentalTheorem} to the entirety of $\starHilbCategory$ should be a priority in future developments, as it would establish a uniformly tight relationship between our framework and more traditional approaches to quantum mechanics and quantum field theory. In the same spirit of relating to mainstream works, we endeavour to explicitly construct more quantum systems of widespread interest---such as wavefunctions/fields over locally compact groups---and to explore more sophisticated applications to quantum field theory and quantum gravity (e.g. constructing analogues of algebraic quantum field theory and introducing Feynman diagrams).

On a different note, we believe that it would be extremely interesting to analyse the natural infinite-dimensional extension of a number of quantum protocols already formalised in CQM. Examples might include simple protocols---such as quantum teleportation and quantum key distribution---or more elaborate protocols---such as the generalised Mermin-type non-locality arguments of Ref. \cite{Gogioso2017c}, and infinite-dimensional extensions of the work on tight reference-frame-independent quantum teleportation of Ref. \cite{Verdon2016}. To kick-start this line of research, an application of $\starHilbCategory$ to the Hidden Subgroup Problem for the infinite group $\integers^n$ has already appeared in Ref. \cite{Gogioso2017b}, based on the original $\starHilbCategory$ from Ref. \cite{Gogioso2016b}.

\newpage
\subparagraph*{Acknowledgements.}
The authors would like to thank Samson Abramsky, Bob Coecke, David Reutter and Masanao Ozawa for comments, suggestions and useful discussions, as well as Sukrita Chatterji and Nicol\`o Chiappori for their support. SG gratefully acknowledges funding from EPSRC and the Williams Scholarship offered by Trinity College. FG gratefully acknowledges funding from the AFSOR grant ``Algorithmic and logical aspects when composing meaning''.

\bibliographystyle{eptcs}

\appendix

\section{The non-standard cyclic group \texorpdfstring{$\starIntegersMod{2 \omega+1}$}{of integers modulo infinities}}

The abelian group $\starIntegersMod{2 \omega+1}$ is defined to be the internal set of non-standard integers $\{-\omega,...,+\omega\}$ endowed with $0$ as unit and with the following binary operation $\oplus$ as group multiplication:
\begin{small}
\begin{equation}
k \oplus h := 
\begin{cases}
	k+h & \text{ if } -\omega \leq k+h \leq +\omega \\
	k+h-(2\omega+1) & \text{ if } +\omega < k+h \\
	k+h+(2\omega+1) & \text{ if }  k+h < -\omega 
\end{cases}
\end{equation}
\end{small}
The group $\starIntegersMod{2 \omega + 1}$ has the integers as a subgroup: if $k,h \in \integers$ are standard integers, then certainly $-\omega \leq k+h \leq +\omega$, and hence $k \oplus h = k+h$. More in general, the group $(\starIntegersModPow{2\omega+1}{n},\oplus,\underline{0})$ contains $\integers^n$ as a subgroup, and as a consequence it is a legitimate non-standard extension of the translation group $\integers^n$ of an $n$-dimensional lattice. Furthermore, the group of automorphisms of $\starIntegersModPow{2\omega+1}{n}$ contains the group of automorphisms of $\integers^n$ as a subgroup (rotations and reflections about the origin are the same, but there are more translations of $\starIntegersModPow{2\omega+1}{n}$ than there are of $\integers^n$).

It is not hard to show that $\starIntegersMod{2 \omega + 1} \isom \integers \times C$ for some dense abelian group $C$. As the elements of $C$, we take exactly one representative from each full copy of $\integers$ in $\{-\omega,+\omega\}$, plus a single element representing both the final segment $+\omega-\naturals$ and the initial segment $-\omega+\naturals$. For each full copy of $\integers$, we take the representative to be the zero element of that copy, and in particular we let $0_C := 0$ (the zero element of the standard integers). Furthermore, we imagine the final segment $+\omega - \naturals$ and the initial segment $-\omega + \naturals$ as glued together to form a single copy of $\integers$, and without loss of generality we pick the representative to be $-\omega$ (so that $-\omega$ is the zero element for that that virtual copy of the integers). Given these considerations, we can always decompose $k \in \starIntegersMod{2 \omega +1}$ uniquely as $(k',\theta_k)$ in terms of a standard integer component $k' \in \integers$ and a representative $\theta_k \in C$:
\begin{small}
\begin{equation}
k' = 
\begin{cases}
	k-\theta_k & \text{ if } \theta_k \neq -\omega\\
	k+\omega & \text{ if } \theta_k = -\omega \text{ and } k \in -\omega + \naturals \\
	k-(\omega+1) & \text{ if } \theta_k = -\omega \text{ and } k \in +\omega - \naturals \\
\end{cases}
\end{equation}
\end{small}
The standard integer components can be added independently of the representatives in $C$ (with some care taken for the boundary case of $\theta_k = -\omega$), so that this defines a group isomorphism $\starIntegersMod{2 \omega+1} \isom \integers \times C$. Also, recall that the infinite positive and negative integers form two dense, uncountable sets, each having no maximum or minimum, and with the finite integers $\integers$ in between \cite{Robinson1974}: as a consequence, the group $C$ is dense and uncountable (but, unlike the non-standard integers, it is not totally ordered).


\noindent There are three different embeddings of the periodic non-standard cubic lattice $\starIntegersModPow{2 \omega + 1}{n}$ that are of interest in this work, reflecting distinct applications to the modelling of cubic lattices $\integers^n$, the approximation of real space $\reals^n$, and the approximation of real toroidal space $\torusGroup{n}$:
\begin{enumerate}
	\item[(i)] the embedding as the lattice $\starIntegersModPow{2 \omega + 1}{n}$ in $\starReals^n$, where we send $\underline{k} \in \starIntegersModPow{2 \omega + 1}{n}$ to $\underline{k} \in \starReals^n$ (not a subgroup);
	\item[(ii)] the embedding as the lattice $\frac{1}{\omega_{uv}}\starIntegersModPow{2 \omega + 1}{n}$ in $\starReals^n$, where $\omega := \omega_{uv}\omega_{ir}$ for some infinite $\omega_{uv},\omega_{ir} \in \starNaturals^+$ and we send $\underline{k} \in\starIntegersModPow{2 \omega + 1}{n}$ to $\underline{p}:= \underline{k}/\omega_{uv} \in \starReals^n$ (also not a subgroup);
	\item[(iii)] the embedding as the subgroup $\frac{1}{2\omega+1}\starIntegersModPow{2 \omega + 1}{n}$ in $\nonstd{\torusGroup{n}}$, where we send $\underline{k} \in \starIntegersModPow{2 \omega + 1}{n}$ to $\frac{1}{2\omega+1}\underline{k} \in \nonstd{\torusGroup{n}}$.
\end{enumerate}
The first embedding uses $\starIntegersModPow{2 \omega + 1}{n}$ to approximate~$\integers^n$, under the observation that the latter is a subgroup of the former. The second embedding instead uses $\starIntegersModPow{2 \omega + 1}{n}$ to approximate $\reals^n$: this is a bit more complicated, as $\reals^n$ cannot be seen as a subgroup of $\starIntegersModPow{2 \omega + 1}{n}$ (the latter is discrete, while the former is dense). However, we can consider the subgroup of $\starIntegersModPow{2 \omega + 1}{n}$ given by those $\underline{k}$ such that $\underline{p} := \underline{k}/\omega_{uv}$ is a near-standard vector in $\starReals^n$, and we can quotient it by infinitesimal equivalence of vectors to obtain the group $\reals^n$. Hence the second embedding can be seen to approximate  $\reals^n$ by using a non-standard lattice of infinitesimal mesh, and working up to infinitesimal equivalence. The third embedding is used similarly to the second embedding, but to approximate $\torusGroup{n}$ instead of $\reals^n$ (with a quotient group homomorphism $\frac{1}{2\omega+1}\starIntegersModPow{2 \omega + 1}{n} \epim \torusGroup{n}$).

It should be noted that $\frac{1}{\omega_{uv}}\starIntegersModPow{2\omega+1}{n}$ is a lattice, and as such it does not enjoy the same symmetries of the continuum $\reals^n$: finite translations can be approximated up to infinitesimals, but rotations cannot. This is in contrast to the $\starIntegersModPow{2\omega+1}{n}$ case, the automorphisms of which contain the automorphisms of $\integers^n$ as a subgroup. When working with quantum systems, however, we are not really interested in the symmetries $\Phi$ of $\reals^n$, but rather in the unitary automorphisms $U_\Phi$ of $\Ltwo{\reals^n}$ that they induce: because the subspace defined by the truncating projector $P_\realsAlgebra{n}$ spans the near-standard vectors, all these unitaries lift from $\Ltwo{\reals^n}$ to $\realsAlgebra{n}$. From the point of view of the non-standard quantum system $\realsAlgebra{n}$, it is \inlineQuote{as if} $\frac{1}{\omega_{uv}}\starIntegersModPow{2\omega+1}{n}$ really possessed all the symmetries of $\reals^n$.

\vspace{-2mm}

\section{Delta functions, plane waves and position/momentum cut-offs}

\vspace{-0.5mm}

We present here the proof that the position eigenstates $\ket{\delta_{\underline{x}}}$ defined in Section \ref{section_box} are the classical states\footnote{Here we only give the proof for the copy condition of $\hbox{\begin{tikzpicture} [scale=1.2,transform shape] 

\def\deltax{0.3} 
\def\deltay{0.5} 


\node [dot, fill=\Dcolour] (mult) at (0,0) {};

\end{tikzpicture}}\!\!$-classical states, but the proofs for the delete and transpose conditions defining $\hbox{\begin{tikzpicture} [scale=1.2,transform shape] 

\def\deltax{0.3} 
\def\deltay{0.5} 


\node [dot, fill=\Dcolour] (mult) at (0,0) {};

\end{tikzpicture}}\!\!$-classical states follow similar lines.} for the group algebra $\hbox{\begin{tikzpicture} [scale=1.2,transform shape] 

\def\deltax{0.3} 
\def\deltay{0.5} 


\node [dot, fill=\Dcolour] (mult) at (0,0) {};

\end{tikzpicture}}\!\!$ of $\starIntegersModPow{2\omega+1}{n}$, exactly when $\underline{x} \in \frac{1}{2\omega+1}\starIntegersModPow{2\omega+1}{n}$:
\begin{align}
	\!\hbox{\begin{tikzpicture} [scale=1.2,transform shape,rotate=-90] 

\def\deltax{0.3} 
\def\deltay{0.5} 


\node (mult_label_outl) at (-\deltax,+\deltay) {};
\node (mult_label_outr) at (+\deltax,+\deltay) {};
\node [dot, fill=\Dcolour] (mult) at (0,0) {};
\node (mult_label_in) at (0,-\deltay) {};
\draw[-] [in=270,out=135] (mult) to (mult_label_outl);
\draw[-] [in=270,out=45] (mult) to (mult_label_outr);
\draw[-] (mult_label_in) to (mult);

\end{tikzpicture}}\!\! \circ \big( \ket{\delta_{\underline{x}}}\big) &= \sum_{\underline{n}} \sum_{\underline{k}} \ket{\goodchi_{\underline{k}}}\otimes \ket{\goodchi_{\underline{n}\ominus\underline{k}}}\; \braket{\goodchi_{\underline{n}}}{\delta_{\underline{x}}}  
	= \sum_{\underline{n}} \sum_{\underline{k}} \ket{\goodchi_{\underline{k}}}\otimes \ket{\goodchi_{\underline{n}\ominus\underline{k}}}\; \goodchi_{\underline{n}}(\underline{x})^\ast  \nonumber \\
	\hspace{2cm}&= \sum_{\underline{n}} \sum_{\underline{k}} \ket{\goodchi_{\underline{k}}} \otimes\ket{\goodchi_{\underline{n}\ominus\underline{k}}}\; \goodchi_{\underline{k}}(\underline{x})^\ast \goodchi_{\underline{n}\ominus\underline{k}}(\underline{x})^\ast e^{i2\pi (2\omega+1)\underline{s}\cdot\underline{x} }  \nonumber \\
	&= \Big[\sum_{\underline{n'}} \goodchi_{\underline{n'}}(\underline{x})^\ast \ket{\goodchi_{\underline{n'}}}\Big] \otimes \Big[\sum_{\underline{k}}  \goodchi_{\underline{k}}(\underline{x})^\ast\ket{\goodchi_{\underline{k}}}\Big]=  \ket{\delta_{\underline{x}}} \otimes  \ket{\delta_{\underline{x}}}.
\end{align}
In the third line, the extra phase $e^{i2\pi (2\omega+1) \underline{s}\cdot\underline{x}}$ appears because $\goodchi_{\underline{k}}$ is a character of $\integers^n$, not of $\starIntegersModPow{2\omega+1}{n}$: the value of $\underline{s} \in \{-1,0,+1\}^n$ keeps track of whether some modular reductions were necessary to go from $\underline{k} \oplus (\underline{n}\ominus \underline{k})$ to $\underline{n}$. It is cancelled out if and only if we require $\underline{x}$ to be in the form $\underline{x} = \underline{j} \frac{1}{2\omega+1}$, for some $\underline{j} \in \starIntegersModPow{2\omega+1}{n}$. Hence the duality between the large-scale cut-off for momenta and the small-scale cut-off for positions, a well-understood phenomenon in quantum mechanics, arises as a consequence of a purely algebraic requirement in our non-standard framework. Similar phases appear for the setups of Section \ref{section_lattice} and Section \ref{section_reals}, with similar large/small-scale dualities following from the algebraic requirement of copiability for classical states.

\end{document}